\documentclass[11pt]{article}
\usepackage[bookmarks=false]{hyperref}

\usepackage{amsmath, amssymb}
\usepackage{amsthm}
\usepackage{complexity}
\usepackage[margin=1in]{geometry}
\usepackage{enumitem}
\usepackage{array}
\usepackage{xspace}
 \usepackage{relsize}
 
\usepackage{thmtools}
\usepackage{thm-restate}

\usepackage{xy}
\xyoption{matrix}
\xyoption{frame}
\xyoption{arrow}
\xyoption{arc}

\usepackage{ifpdf}
\ifpdf
\else
\PackageWarningNoLine{Qcircuit}{Qcircuit is loading in Postscript mode.  The Xy-pic options ps and dvips will be loaded.  If you wish to use other Postscript drivers for Xy-pic, you must modify the code in Qcircuit.tex}
\xyoption{ps}
\xyoption{dvips}
\fi

\entrymodifiers={!C\entrybox}

\newcommand{\bra}[1]{{\left\langle{#1}\right\vert}}
\newcommand{\ket}[1]{{\left\vert{#1}\right\rangle}}
\newcommand{\qw}[1][-1]{\ar @{-} [0,#1]}
\newcommand{\qwx}[1][-1]{\ar @{-} [#1,0]}

\newcommand{\gate}[1]{*+<.6em>{#1} \POS ="i","i"+UR;"i"+UL **\dir{-};"i"+DL **\dir{-};"i"+DR **\dir{-};"i"+UR **\dir{-},"i" \qw}

\newcommand{\control}{*!<0em,.025em>-=-<.2em>{\bullet}}

\newcommand{\ctrl}[1]{\control \qwx[#1] \qw}

\newcommand{\targ}{*+<.02em,.02em>{\xy ="i","i"-<.39em,0em>;"i"+<.39em,0em> **\dir{-}, "i"-<0em,.39em>;"i"+<0em,.39em> **\dir{-},"i"*\xycircle<.4em>{} \endxy} \qw}

\newcommand{\multigate}[2]{*+<1em,.9em>{\hphantom{#2}} \POS [0,0]="i",[0,0].[#1,0]="e",!C *{#2},"e"+UR;"e"+UL **\dir{-};"e"+DL **\dir{-};"e"+DR **\dir{-};"e"+UR **\dir{-},"i" \qw}
\newcommand{\ghost}[1]{*+<1em,.9em>{\hphantom{#1}} \qw}

\newcommand{\lstick}[1]{*!R!<.5em,0em>=<0em>{#1}}

\newcommand{\Qcircuit}{\xymatrix @*=<0em>}

\theoremstyle{plain}
\declaretheorem[name=Theorem]{theorem}
\declaretheorem[name=Lemma, numberlike=theorem]{lemma}
\declaretheorem[name=Corollary, numberlike=theorem]{corollary}

\declaretheorem[name=Definition, numberlike=theorem]{definition}

\declaretheorem[name=Lemma~\ref{lem:toff_proof},numbered=no]{repeated}
\declaretheorem[name=Theorem~\ref{thm:main_real_orthogonal},numbered=no]{real_repeated}

\newcommand{\Per}{\operatorname{Per}}
\newcommand{\encoder}{E}
\newcommand{\decoder}{D}
\newcommand{\NS}{\operatorname{NS_1}}
\newcommand{\B}{R_{\pi/4}}
\renewcommand{\R}{R_{\pi}}
\newcommand{\CCR}{\operatorname{CC-}\!\R}

\newcommand{\CSIGN}{\operatorname{CSIGN}}
\newcommand{\CNOT}{\operatorname{CNOT}}

\newcommand{\sharpP}{\#\P}
\newcommand{\ModpP}{\ComplexityFont{Mod}_{p}\P}
\newcommand{\ptf}{$\varphi$-transition formula\xspace}

\newcommand{\parityP}{\oplus\P}

\renewcommand{\O}{\mathcal{O}}
\renewcommand{\(}{\left(}
\renewcommand{\)}{\right)}

\newcommand\blfootnote[1]{%
  \begingroup
  \renewcommand\thefootnote{}\footnote{#1}%
  \addtocounter{footnote}{-1}%
  \endgroup
}

\title{New Hardness Results for the Permanent Using Linear Optics}
\author{Daniel Grier\thanks{MIT. \ Email: grierd@mit.edu. \ Supported by an NSF
Graduate Research Fellowship under Grant No. 1122374.}
\and Luke Schaeffer\thanks{MIT. \ Email: lrs@mit.edu. } \blfootnote{Both authors were supported by the Vannevar Bush Faculty Fellowship from the US Department of Defense. Part of this research was completed while visiting UT Austin.} }
\date{}

\begin{document}

\maketitle
\begin{abstract}
In 2011, Aaronson gave a striking proof, based on quantum linear optics, that the problem of computing the permanent of a matrix is $\sharpP$-hard.  Aaronson's proof led naturally to hardness of approximation results for the permanent, and it was arguably simpler than Valiant's seminal proof of the same fact in 1979.  Nevertheless, it did not show $\sharpP$-hardness of the permanent for any class of matrices which was not previously known.   In this paper, we present a collection of \emph{new} results about matrix permanents that are derived primarily via these linear optical techniques.   

First, we show that the problem of computing the permanent of a real orthogonal matrix is $\sharpP$-hard. Much like Aaronson's original proof, this implies that even a multiplicative approximation remains $\sharpP$-hard to compute.  The hardness result even translates to permanents of orthogonal matrices over the finite field $\mathbb{F}_{p^4}$ for $p \neq 2, 3$.  Interestingly, this characterization is tight:  in fields of characteristic 2,  the permanent coincides with the determinant; in fields of characteristic 3, one can efficiently compute the permanent of an orthogonal matrix by a nontrivial result of Kogan.

Finally, we use more elementary arguments to prove $\sharpP$-hardness for the permanent of a positive semidefinite matrix.  This result shows that certain probabilities of boson sampling experiments with thermal states are hard to compute exactly, despite the fact that they can be efficiently sampled by a classical computer.  
\end{abstract}

\section{Introduction}
\label{sec:introduction}

The permanent of a matrix has been a central fixture in computer science ever since Valiant showed how it could efficiently encode the number of satisfying solutions to classic \NP-complete constraint satisfaction problems \cite{valiant}.  His theory led to the formalization of many counting classes in complexity theory, including $\sharpP$.  Indeed, the power of these counting classes was later demonstrated by Toda's celebrated theorem, which proved that every language in the polynomial hierarchy could be computed in polynomial-time with only a single call to a $\sharpP$ oracle \cite{toda}.

Let us recall the definition of the matrix permanent.  Suppose $A = ( a_{i,j} )$ is an $n \times n$ matrix over some field.  The \emph{permanent} of $A$ is
$$\Per(A) = \sum_{\sigma \in S_n} \prod_{i=1}^n a_{i, \sigma(i)}$$
where $S_n$ is the group of permutations of $\{ 1, 2, \ldots, n \}$. Compare this to the \emph{determinant} of $A$:
$$\det(A) = \sum_{\sigma \in S_n} \text{sign}(\sigma) \prod_{i=1}^{n} a_{i, \sigma(i)}.$$
Since we can compute the determinant in polynomial time (in fact, in $\NC^2$; see Berkowitz \cite{berkowitz:1984}), the apparent difference in complexity between the determinant and permanent comes down to the cancellation of terms in the determinant \cite{sengupta:1997, valiant:1979}.  

In his original proof, Valiant \cite{valiant} casts the permanent in a combinatorial light, in terms of a directed graph rather than as a polynomial.  Imagine the matrix $A$ encodes the adjacency matrix of a weighted graph with vertices labeled $\{ 1, \ldots, n \}$. Each permutation $\sigma$ on the vertices has a cycle decomposition, which partitions the vertices into a collection of cycles known as a \emph{cycle cover}. The \emph{weight} of a cycle cover is the product of the edge weights of the cycles (i.e., $\prod_{i=1}^{n} a_{i, \sigma(i)}$).  Therefore, the permanent is the sum of the weights of all cycle covers of the graph. Equipped with this combinatorial interpretation of the permanent, Valiant constructs a graph by linking together different kinds of gadgets in such a way that some cycle covers correspond to solutions to a CNF formula, and the rest of the cycle covers cancel out. 

Valiant's groundbreaking proof, while impressive, is fairly opaque and full of complicated gadgets.  A subsequent proof by Ben-Dor and Halevi \cite{ben-dor:1993} simplified the construction, while still relying on the cycle cover interpretation of the permanent.  In 2009, Rudolph \cite{rudolph:2009} noticed an important connection between quantum circuits and matrix permanents---a version of a correspondence we will use often in this paper.  Rudolph cast the cycle cover arguments of Valiant into more physics-friendly language, which culminated in a direct proof that the amplitudes of a certain class of universal quantum circuits were proportional to the permanent.  Had he pointed out that one could embed $\sharpP$-hard problems into the amplitudes of a quantum circuit, then this would have constituted a semi-quantum proof that the permanent is $\sharpP$-hard.  Finally, in 2011, Aaronson \cite{aar:per} (independently from Rudolph) gave a completely self-contained and quantum linear optical proof that the permanent is $\sharpP$-hard. 

One must then ask, what is gained from converting Valiant's combinatorial proof to Aaronson's linear optical one?  One advantage is pragmatic---much of the difficulty of arguments based on cycle cover gadgets is offloaded onto central, well-known theorems in linear optics and quantum computation.  In this paper, we show that the linear optical approach has an even more important role in analyzing permanents of matrices with a global group structure.  Such properties can be very difficult to handle in the ``cycle cover model." For instance, the matrices which arise from Valiant's construction may indeed be invertible, but this seems to be more accidental than intentional, and a proof of their invertibility appears nontrivial.   Adapting such techniques to give hardness results for orthogonal matrices would be extraordinarily tedious.  In contrast, using the linear optical framework, we give proofs of hardness for many such matrices.

This gives a clean example for which a quantum mechanical approach sheds light on a problem in classical theoretical computer science.  To take another example of such a quantum-classical connection, Kuperberg \cite{kuperberg:jones} shows that computing certain values of the Jones polynomial to high accuracy is $\sharpP$-hard using $\PostBQP = \PP$, a well-known result of Aaronson \cite{aar:pp}.  For a more thorough treatment of this topic, see the survey on quantum proofs for classical theorems of Drucker and de Wolf \cite{ddw}.


\subsection{Results}  
We refine Aaronson's linear optical proof technique and show that it \emph{can} provide new $\sharpP$-hardness results.  First, let us formally define what we mean by $\sharpP$-hardness throughout this paper. We say that the \emph{permanent is $\sharpP$-hard} for a class of matrices if all functions in $\sharpP$ can be efficiently computed with single-call access to an oracle which computes permanents of matrices in that class.  That is, the permanent is hard for a function class $A$ if, given an oracle $\O$ for the permanent, $A \subseteq \FP^{\O[1]}$. 

Our main result is a linear optical proof that the permanent of a real orthogonal matrix is $\sharpP$-hard.  Consequently, the permanent of matrices in any of the classical Lie groups (e.g., invertible matrices, unitary matrices, symplectic matrices) is also $\sharpP$-hard.  

Our approach also reveals a surprising connection between the hardness of the permanent of orthogonal matrices over finite fields and the characteristic of the field.  First notice that in fields of characteristic 2, the permanent is equal to the determinant and is therefore efficiently computable.  Over fields of characteristic 3, there exists an elaborate yet polynomial time algorithm of Kogan \cite{kogan:1996} that computes the (orthogonal) matrix permanent.  We give the first explanation for why no equivalent algorithm was found for the remaining prime characteristics, establishing a sharp dichotomy theorem:  for fields of characteristic $2$ or $3$ there is an efficient procedure to compute orthogonal matrix permanents, and for \emph{all} other primes $p$ there exists a finite field\footnote{We prove that this field is $\mathbb{F}_{p^4}$, although in some cases $\mathbb{F}_{p^2}$ or even $\mathbb{F}_{p}$ will suffice.  See Section~\ref{sec:finite_field} for more details.} of characteristic $p$ for which the permanent of an orthogonal matrix (over that field) is as hard as counting the number of solutions to a CNF formula mod $p$.\footnote{Formally, this language is complete for the class $\ModpP$. By Toda's theorem, we have that $\PH \subseteq \BPP^{\ModpP}$.  See Appendix~\ref{sec:complexity} for a more precise exposition of such counting classes.}  Furthermore, there exist infinitely many primes for which computing the permanent of an orthogonal matrix over $\mathbb F_p$ (i.e., modulo $p$) is hard.

Finally, we give a polynomial interpolation argument showing that the permanent of a positive semidefinite matrix is $\sharpP$-hard.  This has an interesting consequence due a recent connection between matrix permanents and boson sampling experiments with thermal input states \cite{chakhmakhchyan:2016, rahimi:2015}.  In particular, the probability of a particular experimental outcome is proportional to a positive semidefinite matrix which depends on the temperatures of the thermal states.  Our result implies that it is hard to compute such output probabilities exactly despite the fact that an efficient classical sampling algorithm exists \cite{rahimi:2015}.

\subsection{Proof Outline} 
The main result concerning the $\sharpP$-hardness of real orthogonal permanents follows from three major steps:
\begin{enumerate}
\item Construct a quantum circuit (over qubits) with the following property:  \emph{If} you could compute the probability of measuring the all-zeros state after the circuit has been applied to the all-zeros state, then you could calculate some $\sharpP$-hard quantity.   We must modify the original construction of Aaronson \cite{aar:per}, so that all the gates used in this construction are real.
\item Use a modified version of the Knill, Laflamme, Milburn protocol \cite{klm} to construct a linear optical circuit which simulates the quantum circuit in the previous step.  In particular, we modify the protocol to ensure that the linear optical circuit starts and ends with one photon in every mode.  Notice that this is distinct from Aaronson's approach \cite{aar:per} because we can no longer immediately use the dual-rail encoding of KLM.  We build new postselected encoding and decoding gadgets to circumvent this problem.
\item Use a known connection (first pointed out by Caianiello \cite{caianiello}) between the transition amplitude of a linear optical circuit and the permanent of its underlying matrix.  Because we paid special attention to the distribution of photons across the modes of our linear optical network in the previous step, the success probability of the linear optical circuit is exactly the permanent of the underlying transition matrix.  It is then simple to work backwards from this permanent to calculate our original $\sharpP$-hard quantity.
\end{enumerate}

The paper is organized as follows. Section~\ref{sec:linear_optics} gives a brief introduction to the linear optical framework and the relevant tools we use in this paper.  In Section~\ref{sec:real_orthogonal}, we use this framework to show that the permanent of a real orthogonal matrix is $\sharpP$-hard. A careful analysis in Section~\ref{sec:finite_field} (and Appendix~\ref{app:fftech}) extends these gadgets to finite fields.\footnote{As is the case with Aaronson's proof, our real orthogonal construction also leads naturally to hardness of approximation results, which we discuss in Appendix~\ref{sec:approx}. }  Finally, in Section~\ref{sec:misc_permanents}, we explore other matrix classes, culminating in a proof that the permanent of a real special orthogonal symplectic involution is $\sharpP$-hard.


\section{Linear Optics Primer} 
\label{sec:linear_optics}

In this section we will introduce the so-called ``boson sampling" model of quantum computation, which will make clear the connection between the dynamics of \emph{noninteracting} bosons and the computation of matrix permanents \cite{caianiello, troyanskytishby}. The most promising practical implementations of this model are based on linear optics and use photons controlled by optical elements such as beamsplitters. We will use the term ``linear optics" throughout, although any type of indistinguishable bosons would have the same dynamics.

Let us first consider the dynamics of a single boson.  At any point in time, it is in one of finitely many \emph{modes}.  As the system evolves, the particle moves from one of, say, $m$ initial modes to a superposition of $m$ final modes according to a \emph{transition matrix} of \emph{amplitudes}. That is, there is an $m \times m$ unitary transition matrix $U \in \mathbb C^{m \times m}$, where $U_{ji}$ is the amplitude of a particle going from mode $i$ to mode $j$.

The model becomes more complex when we consider a system of multiple particles evolving on the same modes according to the same transition matrix.  Let us define states in our space of $k$ bosons in what is called the Fock basis.  A \emph{Fock} state for a $k$-photon, $m$-mode system is of the form $\ket{s_1, s_2, \ldots, s_m}$ where $s_i \ge 0$ is the number of bosons in the $i$th mode  and $\sum_{i=1}^m s_i = k$.  Therefore, the Hilbert space which spans the Fock basis states $\Phi_{m,k}$ has dimension $\binom{k+m-1}{k}$.  Let $\varphi$  be the transformation which lifts the unitary $U$ to act on a multi-particle system.  On a $k$-particle system, $\varphi(U)$ is a linear transformation from $\Phi_{m,k}$ to $\Phi_{m,k}$. 

Let $\ket{S} = \ket{s_1, s_2, \ldots, s_m}$ be the Fock state describing the starting state of the system, and let $\ket{T} = \ket{t_1, t_2, \ldots, t_m}$ be the ending state.  We have:  
$$\bra{T} \varphi(U) \ket{S} = \frac{\Per(U_{S,T})}{\sqrt{s_1! \ldots s_m! t_1! \ldots t_m!}}$$
where $U_{S,T}$ is the matrix obtained by taking $s_i$ copies of the $i$th row and $t_i$ copies of the column $i$ in $U$ for all $i \in \{1, 2, \ldots, m\}$.  We will refer to this formula as the \emph{\ptf}. \footnote{A different intuition, perhaps more appealing for non-physicists, is to think of Fock states as monomials. Specifically, we equate $\ket{s_1, \ldots, s_m}$ with the monomial $x_1^{s_1} \cdots x_m^{s_m}$. That is, there is a variable for each mode, the degree is precisely the number of particles, and the variables commute because bosons are indistinguishable. The unitary acts on polynomials by mapping each variable to a linear combination of variables, e.g., $x_1 \mapsto U_{11} x_1 + \cdots + U_{1m} x_m$. It is not too hard to see that $\Per(U_{S,T})$ is summing up all the ways we can expand $U x^{S}$ to get the monomial $x^{T}$, except that it sums over all permutations, including some permutations which fix $x^{T}$ or $x^{S}$. The normalization factor $1 / \sqrt{s_1! \ldots s_m! t_! \ldots t_m!}$ is there to account for those symmetries. The square root arises due to the fact that we are computing amplitudes, not probabilities. See Aaronson and Arkhipov \cite{aark} for a more thorough treatment of Fock states and polynomials.}
Notice that $s_1 + \cdots + s_m$ must equal $t_1 + \cdots + t_m$ in order for $U_{S,T}$ to be square. This expresses the physical principle that photons are not created or destroyed in the experiment. 

For example, suppose $U$ is the Hadamard gate and that we wish to apply $U$ to two modes each with a single photon.  That is, $U = \frac{1}{\sqrt{2}} \(\begin{smallmatrix} 1 & 1 \\ 1 & -1 \end{smallmatrix}\)$ and $\ket{S} = \ket{1,1}$.  Since the number of photons must be conserved, the resulting state of the system is in some linear combination of $\ket{2,0}, \ket{1,1},$ and $\ket{0,2}$.  We calculate these amplitudes explicitly below:

\begin{center}
\begin{tabular}{l | c c c}
$\ket{T}$ & $\ket{2,0}$ & $\ket{1,1}$ & $\ket{0,2}$  \\ \hline
$U_{S,T}$ & $\frac{1}{\sqrt{2}}\(\begin{smallmatrix} 1 & 1 \\ 1 & 1 \end{smallmatrix}\)$ & $\frac{1}{\sqrt{2}}\(\begin{smallmatrix} 1 & 1 \\  1 & -1 \end{smallmatrix}\)$ & $\frac{1}{\sqrt{2}}\(\begin{smallmatrix} 1 & 1 \\ -1 & -1 \end{smallmatrix}\)$ \rule{0pt}{2.6ex} \\
$\Per(U_{S,T})$ & $1$ & $0$ & $-1$ \\
$\bra{T} \varphi(U) \ket{S}$ \rule{0pt}{2.6ex} & $1/\sqrt{2}$ & $0$ & $-1/\sqrt{2}$ \\
\end{tabular}
\end{center}

Therefore, when we apply Hadamard in a linear optical circuit to the state $\ket{1,1}$ we get the state $\frac{\ket{2,0} - \ket{0,2}}{\sqrt{2}}$.  Indeed, we have derived the famous Hong-Ou-Mandel effect---the photons are noninteracting, yet the final state is clearly highly entangled \cite{hom}.

Finally, we note that $\varphi$ expresses the fact that linear optical systems are reversible and can be composed together. This behavior is captured by the following theorem.

\begin{theorem}[see, for example, Aaronson and Arkhipov \cite{aar:per, aark}]
The map $\varphi$ is a group homomorphism.  Furthermore, if $U \in \mathbb{C}^{n \times n}$ is unitary, then $\varphi(U)$ is unitary.
\end{theorem}

We now state a landmark result in linear optics, which connects the dynamics of a linear optical system with those of a traditional quantum circuit over qubits.  Define $\ket{I} = \ket{0,1, \ldots, 0,1}$, the Fock state with a photon in every other mode.

\begin{theorem}[Knill, Laflamme, and Milburn \cite{klm}]
\label{thm:klm}
Postselected linear optical circuits are universal for quantum computation.  Formally, given a quantum circuit $Q$ consisting of CSIGN and single-qubit gates, there exists a linear optical network $U$  constructible in polynomial time such that 
$$\bra{I} \varphi(U) \ket{I} = \frac{1}{4^{\Gamma}} \bra{0 \cdots 0} Q \ket{0 \cdots 0},$$
where $\Gamma$  is the number of CSIGN gates in $Q$.
\end{theorem}

We will refer to the construction of the linear optical network $U$ from $Q$ in Theorem~\ref{thm:klm} as the \emph{KLM protocol}.  It will be helpful to give some idea of its proof here.   First, each qubit of $Q$ is encoded in two modes of $U$ in the classic dual-rail encoding.  That is, the qubit state $\ket{0}$ is encoded by the Fock state $\ket{0,1}$ and the state $\ket{1}$ is encoded by the Fock state $\ket{1,0}$. 

Now suppose $G$ is a single-qubit gate in $Q$.  Using the \ptf, it is not hard to see that applying $G$ to the corresponding pair of dual-rail modes in the linear optical circuit implements the correct single-qubit unitary.  Applying a CSIGN gate is trickier.  The KLM protocol builds the CSIGN gate from a simpler $\NS$ gate, which flips the sign of a single mode if it has 2 photons and does nothing when the mode has 0 or 1 photon. Using two $\NS$ gates one can construct a CSIGN gate (see Figure~\ref{fig:csign_from_ns} in Appendix~\ref{app:ns_approach}).

Unfortunately, the $\NS$ gate cannot be implemented with a straightforward linear optical circuit.  Therefore, some additional resource is required.  The original KLM protocol uses \emph{adaptive measurements}, that is, the ability to measure in the Fock basis in the middle of a linear optical computation and adjust the remaining sequence of gates if necessary.  Intuitively, using adaptive measurements one can apply some transformation and then measure a subset of the modes to ``check'' if the $\NS$ gate was applied.  For simplicity, however, we will assume we have a stronger resource---namely, \emph{postselection}---so we can assume the measurements always yield the most convenient outcome.  Putting the above parts together completes the proof Theorem~\ref{thm:klm}.

%
%
%


\section{Permanents of Real Orthogonal Matrices}
\label{sec:real_orthogonal}

The first class of matrices we consider are the real orthogonal matrices, that is, square matrices $M \in \mathbb R^{n \times n}$ with $M M^{T} = M^{T} M = I$.  This section is devoted to proving the following theorem, which forms the basis for many of the remaining results in this paper.

\begin{theorem}
\label{thm:main_real_orthogonal}
The permanent of a real orthogonal matrix is $\sharpP$-hard.
\end{theorem}

The orthogonal matrices form a group under composition, the \emph{real orthogonal group}, usually denoted $\mathrm{O}(n,\mathbb R)$. This is a subgroup of the unitary group, $\mathrm{U}(n, \mathbb C)$, which is itself a subgroup of the general linear group $\mathrm{GL}(n, \mathbb C)$. Notice then that the hardness result of Theorem~\ref{thm:main_real_orthogonal} will carry over to unitary matrices and invertible matrices.\footnote{See Corollary~\ref{cor:classical_groups} for a complete list of classical Lie groups for which our result generalizes.}

Our result follows the outline of Aaronson's linear optical proof \cite{aar:per} that the permanent is $\sharpP$-hard. In particular, our result depends on the KLM construction \cite{klm}, and a subsequent improvement by Knill \cite{knill:2002}, which will happen to have several important properties for our reduction. 

Let us briefly summarize Aaronson's argument.  Suppose we are given a classical circuit $C$, and wish to compute $\Delta_C$, the number of satisfying assignments minus the number of unsatisfying assignments.  Clearly, calculating $\Delta_C$ is a $\sharpP$-hard problem.  The first thing to notice is that there exists a simple quantum circuit $Q$ such that the amplitude $\bra{0 \cdots 0}Q\ket{0 \cdots 0}$ is proportional to $\Delta_C$. The KLM protocol of Theorem~\ref{thm:klm} implies that there exists a postselected linear optical experiment simulating $Q$.   This results in the following chain which relates $\Delta_C$ to a permanent.

$$\Per(U_{I, I}) = \bra{I} \varphi(U) \ket{I} \propto  \bra{0 \cdots 0} Q \ket{0 \cdots 0} \propto \Delta_C.$$

Notice that Aaronson's result does not imply that the permanent of $U \in \mathrm{U}(n, \mathbb C)$ is $\sharpP$-hard since $U_{I,I}$ is a \emph{submatrix} of $U$.  If, however, $\ket{S} = \ket{T} = \ket{1, \ldots, 1}$, then $U_{S,T} = U$ so the analogous chain relates $\Delta_C$ directly to the permanent of $U$, which is a complex unitary matrix. In fact, this is exactly what we will arrange by modifying the KLM protocol.  Furthermore, we will be careful to use \emph{real} matrices exclusively during all gadget constructions, which will result in $U$ being real, finishing the proof of Theorem~\ref{thm:main_real_orthogonal}.
 

In the following subsections, we will focus on the exact details of the reduction and emphasize those points where our construction differs from that of Aaronson.  

\subsection{Constructing the Quantum Circuit}
\label{subsec:quantum_circuit}
Let $C : \{0,1\}^n \rightarrow \{0,1\}$ be a classical Boolean circuit of polynomial size and let $$\Delta_C :=  \sum_{x \in \{0,1\}^n} (-1)^{C(x)}.$$  In this section, we prove the following:
\begin{theorem}
\label{thm:build_Q}
Given $C$, there exists a $p(n)$-qubit quantum circuit $Q$ such that $$\bra{0}^{\otimes p(n)} Q \ket{0}^{\otimes p(n)} = \frac{\Delta_C}{2^n}$$ where $p(n)$ is some polynomial in $n$.  Furthermore, $Q$ can be constructed in polynomial time with a polynomial number of real single-qubit gates and $\CSIGN$ gates.
\end{theorem}

To prove the theorem, it will suffice to implement $\O_C$, the standard oracle instantiation of $C$ on $n+1$ qubits.  That is, $\O_C \ket{x, b} = \ket{x, b \oplus C(x)}$ for all $x \in \{0,1\}^n$ and $b \in \{0,1\}$.  The circuit for $Q$ is depicted below, where $H$ is the Hadamard gate and $Z = (\begin{smallmatrix} 1 & 0 \\ 0 & -1 \end{smallmatrix})$ is the Pauli $\sigma_Z$ gate.

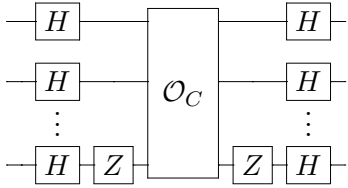
\begin{center}
\mbox{
\Qcircuit @C=.5em @R=.8em {
& \qw & \gate{H} & \qw & \multigate{3}{\O_C} & \qw & \gate{H} & \qw \\ 
& \qw & \gate{H} & \qw & \ghost{\O_C} & \qw & \gate{H} & \qw \\ 
&  & \raisebox{.5em}{\vdots} &  &  &  & \raisebox{.5em}{\vdots}  &  \\ 
& \qw & \gate{H} & \gate{Z} &\ghost{\O_C} & \gate{Z} & \gate{H} & \qw
}}
\end{center}

From this construction, we have 
$$\bra{0}^{\otimes p(n)} Q \ket{0}^{\otimes p(n)} = \frac{1}{2^n} \( \sum_{x \in \{0,1\}^n} \bra{x} \bra{-} \) \O_C \( \sum_{x \in \{0,1\}^n} \ket{x} \ket{-} \) = \frac{\Delta_C}{2^n}.$$

Therefore, to complete the proof, it suffices to construct $\O_C$ from CSIGN and single-qubit gates. For now let us assume we have access to the Toffoli gate as well. Since $C$ is a classical Boolean function of polynomial complexity, $C$ can be implemented with a polynomial number of Toffoli and NOT gates\footnote{Because we require that all ancillas start in the $\ket{0}$ state, we also need the NOT gate to create $\ket{1}$ ancillas.} and a polynomial number of ancillas starting in the $\ket{0}$ state \cite{toffoli}. 

Let us describe, briefly, one way to construct $\O_C$. Suppose we are given the circuit $C$ as a network of polynomially many NAND gates. For each wire, with the exception of the input wires, we create an ancilla initially in state $\ket{0}$ and use the NOT gate to put it in state $\ket{1}$. For each NAND gate (in topological ordering, i.e., such that no gate is applied before its inputs have been computed), we apply a Toffoli gate targeting the ancilla associated with the output wire, and controlled by the qubits associated with its input wires (whether they are the output of an earlier NAND gate, or an actual input). Hence, the target qubit is in state $\ket{1}$ unless both control qubits are in state $\ket{1}$, simulating a NAND gate. Once we have applied all the gates of $C$, the output of the function will exist in the final ancilla register. We can now apply the same sequence of gates (ignoring the final Toffoli gate) in reverse order, which returns all other ancillas and inputs to their original value.  This completes the construction.


Finally, we must construct the Toffoli gate from single-qubit gates and CSIGN gates.  Unfortunately, Aaronson's proof \cite{aar:per} uses a classic construction of the Toffoli gate which uses complex single-qubit gates (see, for example, Nielsen and Chuang \cite{nc}).  This will later give rise to linear optical circuits with complex matrix representations as well.\footnote{Actually, the proof of Aaronson \cite{aar:per} claims that the final linear optical matrix consists entirely of real-valued entries even though the matrices of the individual single-qubit gates have complex entries.  In fact, the matrix \emph{does} have complex entries, but our construction for Toffoli suffices to fix this error.}  Therefore, we will restrict ourselves to CSIGN and \emph{real} single-qubit gates in our construction of the Toffoli gate.\footnote{Although it is known that the Toffoli gate and the set of real single-qubit gates suffice to densely generate the orthogonal matrices (i.e., $\mathrm{O}(2^n)$ for every $n>0$) \cite{shi:gate}, it will turn out to be both simpler and necessary to have an exact decomposition.  In particular, we will need an exact construction of the Toffoli gate in Section~\ref{sec:finite_field} where we discuss the computation of permanents in finite fields.}

\begin{lemma}
\label{lem:toff_proof}
There exists a circuit of CSIGN, Hadamard, and $\B$ gates which implements a Toffoli gate \emph{exactly}, where 
$$\B = \frac{1}{2}
\begin{pmatrix}
 \sqrt{2+\sqrt{2}} & -\sqrt{2-\sqrt{2}} \\
 \sqrt{2-\sqrt{2}} & \sqrt{2+\sqrt{2}} 
\end{pmatrix}.$$
\end{lemma}

We prove this lemma in Appendix~\ref{app:toff_proof}.  This completes the proof of Theorem~\ref{thm:build_Q}.

\subsection{Postselected Linear Optical Gadgets}
\label{subsec:linear_optical}

We will construct a postselected linear optical circuit $O$ which will simulate the qubit circuit $Q$ on the all zeros input via a modified version of the KLM protocol.  The following chain of relations will hold:\footnote{To clarify, $\ket{0 \cdots 0}$ is a tensor product of qubits in the state $\ket{0}$ and $\ket{1, \ldots, 1}$ is a Fock state with $1$ photon in every mode.}
$$\Per(O) = \bra{1, \ldots,  1} \varphi(O) \ket{1, \ldots, 1} \propto  \bra{0 \cdots 0} Q \ket{0 \cdots 0} \propto \Delta_C.$$

The first step was to convert from a classical circuit to a quantum circuit.  Below we formalize the second step: converting from a quantum circuit to a linear optical circuit.

\begin{theorem}
\label{thm:klm_plus}
Given an $n$-qubit quantum circuit $Q$ with a polynomial number of CSIGN and real single-qubit gates, there exists a linear optical circuit $O$ on $4n + 2\Gamma$ modes such that 
$$\bra{1, \ldots, 1} \varphi(O) \ket{1, \ldots, 1}  = \(\frac{1}{3}\sqrt{\frac{2}{3}}\)^\Gamma \(\frac{-1}{\sqrt{6}}\)^n  \bra{0}^{\otimes n} Q \ket{0}^{\otimes n},$$
where $\Gamma$ is the number of CSIGN gates in $Q$.  Furthermore, the matrix representing $O$ is a real orthogonal matrix and can be computed in time polynomial in $n$.
\end{theorem}

We now show how to construct such a linear optical circuit $O$ using the original KLM protocol, subsequent improvements by Knill \cite{knill:2002}, and a new gadget unique to our problem.  First, let us recall our main issue with using the original KLM protocol:  to prove that \emph{orthogonal} matrices are $\sharpP$-hard, we must have that all modes start and end with exactly one photon.  There are two instances in which the original KLM protocol requires a mode to be empty at the beginning and end of the computation.  First, the $\NS$ gate postselects on the Fock state $\ket{0,1}$, and second, KLM protocol works in a dual-rail encoding.  Therefore, half of the modes in the original KLM protocol start and end empty.

To overcome the first obstacle, we appeal to subsequent work of Knill \cite{knill:2002}, in which the $\NS$ gadget construction for CSIGN is replaced by a single $4$-mode gadget $V$, which directly implements CSIGN with two modes postselected in state $\ket{1,1}$.  From the matrix gadget
$$V = \frac{1}{3 \sqrt{2}} \begin{pmatrix}
 -\sqrt{2} & -2 & 2 & 2 \sqrt{2} \\
 2 & -\sqrt{2} & -2 \sqrt{2} & 2 \\
 -\sqrt{6 + 2 \sqrt{6}} & \sqrt{6-2 \sqrt{6}} & -\sqrt{3+\sqrt{6}} & \sqrt{3-\sqrt{6}} \\
 -\sqrt{6-2 \sqrt{6}} & -\sqrt{6 +2 \sqrt{6}} & -\sqrt{3-\sqrt{6}} & -\sqrt{3+\sqrt{6}} \\
\end{pmatrix}
$$
we can directly calculate the transition amplitudes of the circuit:
\begin{align*}
\bra{0, 0, 1, 1} \varphi(V) \ket{0, 0, 1, 1} &= \textstyle{\frac{1}{3}\sqrt{\frac{2}{3}}} & \bra{0, 1, 1, 1} \varphi(V) \ket{1, 0, 1, 1} &= 0 \\
\bra{0, 1, 1, 1} \varphi(V) \ket{0, 1, 1, 1} &= \textstyle{\frac{1}{3}\sqrt{\frac{2}{3}}} &
\bra{1, 0, 1, 1} \varphi(V) \ket{0, 1, 1, 1} &= 0 \\
\bra{1, 0, 1, 1} \varphi(V) \ket{1, 0, 1, 1} &= \textstyle{\frac{1}{3}\sqrt{\frac{2}{3}}} &
\bra{2, 0, 1, 1} \varphi(V) \ket{1, 1, 1, 1} &= 0 \\
\bra{1, 1, 1, 1} \varphi(V) \ket{1, 1, 1, 1} &= \textstyle{-\frac{1}{3}\sqrt{\frac{2}{3}}} &\bra{0, 2, 1, 1} \varphi(V) \ket{1, 1, 1, 1} &= 0
\end{align*}

We now argue that these transition amplitudes suffice to generate a postselected CSIGN.  Consider the linear optical circuit depicted in Figure~\ref{fig:app_V}: the first two inputs of the $V$ gadget are applied to the dual rail modes which contain a photon whenever their corresponding input qubits of the CSIGN gate are in state $\ket{1}$; the next two modes are postselected in the $\ket{1,1}$ state.  First, because we postselect on the final two modes ending in the state $\ket{1,1}$, we only need to consider those transitions for which those two modes end in that state.  Secondly, because we use ``fresh'' ancillary modes for every CSIGN gate, we can always assume that those two modes start in the $\ket{1,1}$ state.  This already vastly reduces the number of cases we must consider.

\begin{figure}[h]
\centering
\mbox{
\Qcircuit @C=1em @R=.6em {
& & \multigate{3}{V} & \qw & & &         & & &  \multigate{1}{\CSIGN} & \qw \\ 
& & \ghost{V} & \qw & & &           & & & \ghost{\CSIGN} & \qw \\
\lstick{\ket{1}} & & \ghost{V} & \qw & \ket{1} & & \raisebox{1em}{=} & & & \qw & \qw  \\
\lstick{\ket{1}} & & \ghost{V} & \qw & \ket{1} & &          & & &    \qw & \qw 
}}
\caption{Applying a postselected $V$ gadget to generate CSIGN.}
\label{fig:app_V}
\end{figure}
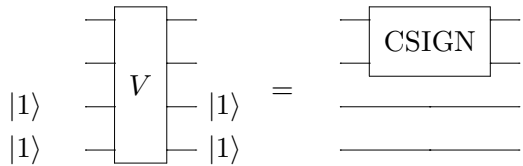

Finally, we wish to know what will happen when the first two modes start in the states $\ket{0,0}, \ket{0,1}, \ket{1,0},$ and $\ket{1,1}$.  Our construction will ensure that there is never more than one photon per mode representing one of the dual-rail encoded qubits.  For instance, the transition amplitudes of $V$ show that whenever the first two modes of the circuit each start with a photon, there is 0 probability (after postselection) that those photons transition to a state in which one of those modes contains 2 photons and the other contains no photons.  

We find that all other amplitudes behave exactly as we would expect for CSIGN.  Since each of the acceptable transitions (e.g. from the state $\ket{0,1}$ to the state $\ket{0,1}$) has equal magnitude, we only have left to check that $V$ flips the sign of the state whenever the input modes are both in the $\ket{1}$ state, which is indeed the case.  Importantly, because $\varphi$ is a homomorphism, we can analyze each such gate separately.  Therefore, using the above we can now construct a linear optical circuit $O$ where all of our postselected modes for CSIGN start and end in the $\ket{1}$ state. 

We now turn our attention to the dual-rail encoding.  Instead of changing the dual-rail encoding of the KLM protocol directly, we will  start with one photon in every mode and apply a linear optical gadget to convert to a dual-rail encoding.  Of course, the number of photons in the circuit must be conserved, so we will dump these extra photons into $n$ modes separate from the modes of the dual-rail encoding.  Specifically, each logical qubit is our scheme is initially represented by four modes in the state $\ket{1,1,1,1}$.  We construct a gadget that moves a photon from the first mode to the third mode, postselecting on a single photon in the last mode.  That is, under postselection, we get the transition
$$\ket{1,1,1,1} \rightarrow \ket{0,1,2,1},$$
where the first two modes now represent a $\ket{0}$ qubit in the dual-rail encoding, the third mode stores the extra photon (which we will reclaim later), and the last qubit is necessary for postselection.  We call the gadget for this task the \emph{encoding gadget} $E$, and it is applied to the first, third, and fourth mode of the above state.  The matrix for $E$ is 
$$\encoder = \frac{1}{\sqrt{6}}
\begin{pmatrix}
 \sqrt{2} & -\sqrt{2} & \sqrt{2} \\
 0 & \sqrt{3} & \sqrt{3} \\
 -2 & -1 & 1 \\
\end{pmatrix}$$
from which we get the following transition amplitudes
\begin{center}
\begin{tabular}{l c l c l }
$\bra{1,1,1} \varphi(\encoder) \ket{1,1,1} = 0$, & & $\bra{2,0,1} \varphi(\encoder)\ket{1,1,1}  = 0$, &  & $\bra{0,2,1} \varphi(\encoder) \ket{1,1,1}= \frac{1}{\sqrt{3}}$. 
\end{tabular} 
\end{center}

After applying the encoding gadget to each logical qubit, we can implement the KLM protocol as previously discussed.\footnote{One might wonder why we cannot simply apply the encoding gadget to the \emph{entire} input, thus circumventing the need to use Knill's more complicated $V$ gadget to implement CSIGN.  Examining Theorem~\ref{thm:klm} carefully, we see that all the postselection actually happens at the end of the computation.  One might be concerned that once we measured the state $\ket{0,1}$ to implement $\NS$, those modes would remain in that state.  Nevertheless, it \emph{is} possible to compose the gadgets in such a way to allow for postselection on $\ket{0}$ while maintaining that the desired amplitude is still on the $\ket{1, \ldots, 1}$ state.  We omit such a design since $V$ will turn out to have some nice properties, including its minimal usage of ancillary modes.}  Therefore, the relevant amplitude in the computation of $Q$ is now proportional to amplitude of the Fock state which has $n$ groups of modes in the state $\ket{1,0,2,1}$ and $2\Gamma$ modes in the state $\ket{1}$.  Because we want to return to a state which has one photon in every mode, we must reverse the encoding step.\footnote{Notice that postselection was required for the encoding gadget, so it does not have a natural inverse.}  For this purpose, we construct a \emph{decoding gadget} $D$, which will not require any extra postselected modes.  We apply the gadget to the second and third modes of the logical qubit such that the two photons in the third mode split with some nonzero probability.  The matrix for $D$ is 
$$\decoder = \frac{1}{\sqrt{2}}
\begin{pmatrix}
1 & 1 \\ 1 & -1 
\end{pmatrix}$$
from which the transition condition $\bra{1,1} \varphi(\decoder) \ket{0,2} = -1/\sqrt{2}$ follows.  Nearly any two-mode linear optical gate would suffice here, but $D$, the familiar Hadamard gate, optimizes the amplitude above.  Notice that if the logical qubit had been in the $\ket{1}$ state, then we would have applied $D$ to a three-photon state.  Because of conservativity, the resulting amplitude on the $\ket{1,1}$ state would be zero.  Putting together the encoding gadget, the KLM scheme, and the decoding gadget completes the proof of the theorem.

\subsection{Main Result}
We are finally ready to prove the Theorem~\ref{thm:main_real_orthogonal}, which we restate below.

\begin{real_repeated}
The permanent of a real orthogonal matrix is $\sharpP$-hard.  Specifically, given a polynomially sized Boolean circuit $C$, there exist integers $a, b \in \mathbb{Z}$ and a real orthogonal matrix $O$ computable in polynomial time such that 
$$\Per(O) = 2^a 3^b \Delta_C.$$
\end{real_repeated}
\begin{proof}
We reduce from the problem of calculating $\Delta_C$ for some polynomially sized Boolean circuit $C$ on $n$ bits.  By Theorem~\ref{thm:build_Q}, we first construct the quantum circuit $Q$ from CSIGN and single-qubit gates such that 
$\bra{0}^{\otimes p(n)} Q \ket{0}^{\otimes p(n)} = \Delta_C/2^n$.  Let $\Gamma$ be the number of CSIGN gates in $Q$.  We then convert the qubit circuit $Q$ to a linear optical circuit $O$ on $4p(n) + 2\Gamma$ modes using Theorem~\ref{thm:klm_plus}.  Notice that we can assume without loss of generality that $p(n)$ and $\Gamma$ are both even since we can always add an extra qubit to the circuit $Q$ and/or apply an extra CISGN gate to the $\ket{00}$ state.  Combined with the fact that the output amplitudes of linear optical experiments can be described by permanents via the \ptf, we have the following chain of consequences
\begin{align*}
\Per(O) &= \bra{1, \ldots, 1} \varphi(O) \ket{1, \ldots, 1} \\
&=  \(\frac{1}{3}\sqrt{\frac{2}{3}}\)^\Gamma \(\frac{-1}{\sqrt{6}}\)^{p(n)} \bra{0}^{\otimes p(n)} Q \ket{0}^{\otimes p(n)} \\
&=  \(\frac{1}{3}\sqrt{\frac{2}{3}}\)^\Gamma \(\frac{-1}{\sqrt{6}}\)^{p(n)} \(\frac{1}{2^n}\) \Delta_C \\
&=   2^a 3^b \Delta_C,
\end{align*}
where the last equality comes from the fact that $\Gamma$ and $p(n)$ are even.
\end{proof}


\section{Permanents over Finite Fields}
\label{sec:finite_field}

Valiant's foundational work on $\sharpP$ is well-known, but his contemporary work on the relationship between the permanent and the class we now know as $\ModkP$ is less appreciated. In another 1979 paper \cite{valiant_completeness_classes}, Valiant showed that the permanent modulo $p$ is $\ModpP$-complete, except when $p = 2$, in which case the permanent coincides with the determinant because $1 \equiv -1 \pmod{2}$. 
\begin{theorem}[Valiant \cite{valiant_completeness_classes}]
The problem of computing $\Per(M) \bmod p$ for a square matrix $M \in \mathbb F_p^{n \times n}$ is $\ModpP$-complete for any prime $p \neq 2$ (and in $\NC^2$ otherwise). 
\end{theorem}
As discussed in Appendix~\ref{sec:complexity}, $\ModpP$-hardness provides evidence for the difficulty of computing the permanent, even modulo a prime. In particular, an efficient algorithm for the problem would collapse the polynomial hierarchy. 

In the spirit of our result on real orthogonal matrices, we ask whether the permanent is still hard for orthogonal matrices in a finite field. We are not the first to consider the problem; there is the following surprising theorem of Kogan \cite{kogan:1996} in 1996. 
\begin{theorem}[Kogan \cite{kogan:1996}]
\label{thm:characteristic3easy}
Let $\mathbb F$ be any field of characteristic $3$. There is a polynomial time algorithm to compute the permanent of any \emph{orthogonal} matrix over $\mathbb F$.
\end{theorem}
In other words, for orthogonal matrices, the permanent is easy to compute for fields of characteristic 2 (since it is easy in general), but it is also easy for fields of characteristic 3 (by a much more elaborate argument)! Could it be that the permanent is easy for all finite fields of some other characteristic? No, it turns out. Using the gadgets from Section~\ref{sec:real_orthogonal}, we prove a converse to Theorem~\ref{thm:characteristic3easy}.

\begin{theorem}
\label{thm:ffmain1}
Let $p \neq 2, 3$ be a prime. There exists a finite field of characteristic $p$, namely $\mathbb F_{p^4}$, such that the permanent of an orthogonal matrix in $\mathbb F_{p^4}$ is $\ModpP$-hard.
\end{theorem}
We prove the theorem by carefully porting Theorem~\ref{thm:main_real_orthogonal} to the finite field setting. Recall that Theorem~\ref{thm:main_real_orthogonal} takes a circuit $C$ and constructs a sequence of gadgets $G_1, \ldots, G_m$ such that 
\begin{equation}
\Per(G_1 \cdots G_m) = 2^{a} 3^{b} \Delta_{C},
\label{eqn:real_orthogonal}
\end{equation}
for some $a, b \in \mathbb Z$. In general, there is no way to convert such an identity on real numbers into one over finite fields, but all of our gadgets are built out of \emph{algebraic} numbers. In particular, all of the entries are in some algebraic field extension $\mathbb Q(\alpha)$ of the rationals, where $\alpha \approx 4.182173283$ is the largest real root of irreducible polynomial
$$
f(x) = x^{16} - 40x^{14} + 572x^{12} - 3736x^{10} + 11782 x^8 - 17816 x^6 + 11324 x^4 - 1832 x^2 + 1.
$$
Each element in $\mathbb Q(\alpha)$ can be written as a polynomial (of degree less than 16) in $\alpha$ over the rationals. In Appendix~\ref{app:fftech_entries}, we give explicit canonical representations for a set of numbers which generate (via addition, subtraction and multiplication, but \emph{not} division) the entries of all our gadgets.

Each entry of a gadget $G_i$ is a polynomial in $\alpha$ with rational coefficients, so observe that we can take a common denominator for the coefficients and write the entry as an integer polynomial divided by some positive integer. By the same token, we can take a common denominator for the entries of a gadget $G_i$, and write it as $\frac{1}{k_i}\hat{G_i}$ where $\hat{G_i}$ is a matrix over $\mathbb Z[\alpha]$, and $k_i$ is a positive integer.

Now we would like to take Equation~\ref{eqn:real_orthogonal} modulo a prime $p$. In principle, we can pull $k_1, \ldots, k_m$ out of the permanent, multiply through by $Z = (k_1 \cdots k_m)^n 2^{|a|} 3^{|b|}$ to remove all fractions on both sides, and obtain an equation of the form 
$$
K \Per(\hat{G_1} \cdots \hat{G_m}) = K' \Delta_C,
$$
where $K, K'$ are integers. Then the entire equation is over $\mathbb Z[\alpha]$, so if we reduce all the coefficients modulo $p$, we get an equation over $\mathbb F_p[\alpha]$.

We show in Appendix~\ref{app:fftech_entries} that for each gadget we use, the denominator $k_i$ may have prime divisors $2$, $3$, and $23$, but no others. Hence, as long as $p \neq 2, 3, 23$ (and in the case $p = 23$, there is an alternative representation we can use, see Appendix~\ref{app:fftech}), we can divide through by $Z$, pull it back inside the permanent as the $\frac{1}{k_i}$s, and distribute each $\frac{1}{k_i}$ into the corresponding $\hat{G_i}$. This gives
$$
\Per(G_1 \cdots G_m) \equiv 2^{a} 3^{b} \Delta_{C} \pmod p,
$$
the equivalent of Equation~\ref{eqn:real_orthogonal}, but over $\mathbb F_p[\alpha]$. In particular, $G_1, \ldots, G_m$ are now orthogonal matrices in $\mathbb F_p[\alpha]$, and $\Delta_C$ has been reduced modulo $p$. 

Note that $\mathbb F_p[\alpha] \cong \mathbb F_p[x] / (f(x))$ is a ring, not a field. If $f(x)$ were irreducible modulo $p$ then it would be a field, but this will never happen for our $f$. Consider the following lemma. 
\begin{lemma}
\label{lem:squaring_roots}
Let $q$ be a prime power.  Suppose $\mathbb F_q$ is the subfield of order $q$ contained in the finite field $\mathbb F_{q^2}$. Then every element in $\mathbb F_q$ has a square root in $\mathbb F_{q^2}$. 
\end{lemma}
\begin{proof}
Let $a$ be an arbitrary element of $\mathbb F_q$. By definition, $a$ has a square root if the polynomial $f(x) := x^2 - a$ has a root. If $f$ has a root in $\mathbb F_q$ then we are done. Otherwise, $f$ is irreducible, but has a root in $\mathbb F_q[x] / \langle f(x) \rangle \cong \mathbb F_{q^2}$. 
\end{proof}
By Lemma~\ref{lem:squaring_roots}, the square roots of $2$ and $6$ are in $\mathbb F_{p^2}$, and therefore so are $2 + \sqrt{2}$ and $3 + \sqrt{6}$. Then \emph{their} square roots are in $\mathbb F_{p^4}$, so $\alpha = \sqrt{2 + \sqrt{2}} + \sqrt{3 + \sqrt{6}}$ is in $\mathbb F_{p^4}$. All the other roots of $f$ can be expressed as polynomials in $\alpha$ (see Appendix~\ref{app:fftech_galois}), so they are all in $\mathbb F_{p^4}$. 
It follows that $f$ factors over $\mathbb F_p$ as a product of irreducible polynomials, each of degree $1$, $2$, or $4$.

Suppose $g$ is some irreducible factor of $f$. The ideal $(g(x))$ contains $(f(x))$, so there exists a ring homomorphism $\sigma$ from $\mathbb F_p[x] / (f(x))$ to $\mathbb F_p[x] / (g(x))$. Note that $\mathbb F_p[x] / (g(x))$ \emph{is} a field because $g(x)$ is irreducible over $\mathbb F_p$. Also, $\sigma$ fixes $\mathbb F_p$, so we obtain 
$$
\Per(\sigma(G_1) \cdots \sigma(G_m)) = \sigma(\Per(G_1 \cdots G_m)) = 2^{a} 3^{b} \Delta_{C}
$$
as an equation over the field $\mathbb F_p[x] / (g(x))$. For each $i$, $\sigma(G_i)$ is orthogonal in $\mathbb F_p[x] / (g(x))$ as well:
$$
\sigma(G_i) \sigma(G_i)^T = \sigma(G_i G_i^T) = \sigma(I) = I.
$$
It follows that $M := \sigma(G_1) \cdots \sigma(G_m)$ is orthogonal. 

Depending on the degree of $g$, the field $\mathbb F_p[x] / (g(x))$ is isomorphic to $\mathbb F_p$, $\mathbb F_{p^2}$, or $\mathbb F_{p^4}$. But $\mathbb F_{p^4}$ contains $\mathbb F_p$ and $\mathbb F_{p^2}$, so $M$ can be lifted to a matrix over $\mathbb F_{p^4}$. Given the permanent of $M$ in $\mathbb F_{p^4}$, we can easily solve for $\Delta_C$, so this completes the proof of Theorem~\ref{thm:ffmain1}. 

Theorem~\ref{thm:ffmain1} shows that for any prime $p \neq 2, 3$ there is some finite field of characteristic $p$ where computing permanents (of orthogonal matrices) is hard. In particular, $p=2$ and $p=3$ are the only cases where the permanent of an orthogonal matrix is easy to compute in \emph{every} finite field of characteristic $p$, assuming the polynomial hierarchy does not collapse. We will now show that there are primes $p$ for which this problem is hard in \emph{any} field of characteristic $p$, by showing that it is hard to compute in $\mathbb F_p$ (which is contained in every other field of characteristic $p$).   

\begin{theorem}
\label{thm:ffmain2}
For all but finitely many primes $p$ that split completely in $\mathbb Q(\alpha)$, computing the permanent of an orthogonal matrix over $\mathbb F_p$ is $\ModpP$-complete. This is a sequence of primes with density $\frac{1}{16}$ beginning
$$191, 239, 241, 337, 383, 433, 673, 863, 911, 1103, 1151, 1249, 1583, 1871, 1873, 2017, \ldots$$
\end{theorem}
\begin{proof}
Recall that in the proof of Theorem~\ref{thm:ffmain1}, if $g$ is an irreducible factor of $f$, then the result applies over the field $\mathbb F_{p}[x] / (g(x)) \cong \mathbb F_{p^{\deg g}}$. We show that $g$ is degree at most $4$, but in special cases this can be improved. In particular, we want $g$ to be degree $1$ (i.e., a linear factor) for our orthogonal matrix to be over $\mathbb F_p$. 

First, observe that $\mathbb Q(\alpha)$ is a Galois extension of $\mathbb Q$. That is, every root of the minimal polynomial for $\alpha$ is in $\mathbb Q(\alpha)$. See Appendix~\ref{app:fftech_galois} for details.  We apply Chebotarev's density theorem \cite{Chebotarev}, which says that if $K$ is a finite Galois extension of $\mathbb Q$ of degree $n$, then the density of primes which split completely in $K$ is $\frac{1}{n}$. We take $K = \mathbb Q(\alpha)$, a degree 16 extension of $\mathbb Q$. 

For our purposes, a prime $p$ splits completely if and only if the ideal $(p)$ factors into 16 distinct maximal ideals in the ring of integers of $\mathbb{Q}(\alpha)$.  For all by finitely many such primes,\footnote{Actually, we can compute these primes explicitly as those that divide the index of $\mathbb{Z}[\alpha]$ in the ring of integers of $\mathbb{Q}(\alpha)$.  For our choice of field, this number is $19985054955504338544361472 = 2^{75} 23^2$.} we also have that $f$ (the minimal polynomial for $\alpha$) factors into distinct linear terms modulo $p$ by Dedekind's theorem. Furthermore, since $\mathbb Q(\alpha)$ is a Galois extension, $f$ will split into equal degree factors. Hence, if any factor is linear, then \emph{all} the factors are linear. 

Therefore, according to Chebotarev's theorem, $(1/16)$th of all primes split completely and yield the desired hardness result. We verified the list of primes given in the theorem computationally. Also note that for about $3/16$ of all primes, we can prove a hardness result over $\mathbb F_{p^2}$ but not $\mathbb F_p$. 
\end{proof}

We leave open how hard it is to compute the permanent of an orthogonal matrix modulo the remaining $\frac{15}{16}$ of all primes. Other linear optical gadgets can be used for $\CSIGN$ instead of $V$, resulting in different field extensions where different primes split. For instance, there exists an orthogonal gadget for KLM's $\NS$ gate for which computing the permanent modulo $97$ is hard (see Appendix~\ref{app:ns_approach}). However, it seems impossible to design linear optical gadgets that do not involve $2$ or $3$ photons at a time, in which case writing down $\varphi(O)$ requires $\sqrt{2}$ and $\sqrt{3}$. By quadratic reciprocity, these square roots only exist if $p \equiv \pm 1 \pmod{24}$ (i.e., for about a quarter of all primes), so the remaining primes may require some other technique.


\section{Expanding Permanent Hardness}
\label{sec:misc_permanents}

In this section, we try to fill in some of the remaining landscape of matrix permanents.  In particular, we will focus on the permanents of positive semidefinite (PSD) matrices and their connection to boson sampling.  We will conclude by listing some matrix variants and their accompanying permanent complexities, many of which are simple consequences of the reduction in Section~\ref{sec:real_orthogonal}.

\subsection{Positive Semidefinite Matrix Permanents}
Permanents of PSD matrices have recently become relevant to the expanding theory of boson sampling \cite{rahimi:2015}.  Namely, permanents of PSD matrices describe the output probabilities of a boson sampling experiment in which the input is a tensor product of thermal states.  That is, each of the $m$ modes of the system starts in a state of the form 
$$\rho_i = (1 - \tau_i) \sum_{n=0}^\infty \tau_i^n \ket{n}\bra{n}$$
where $\tau_i = \left<n_i\right>/( \left<n_i\right> + 1)$ and $ \left<n_i\right>$ is average number of photons one observes when measuring $\rho_i$.  In particular, notice that $\tau_i \ge 0$.

Let $U$ be a unitary matrix representing a linear optical network.  Let $D$ be the diagonal matrix with $\tau_1, \ldots, \tau_m$ along the diagonal, and $A := U D U^\dag$.  Since $\tau_i \ge 0$ for all $i$, $A$ is PSD.  We can calculate the probability of detecting one photon in each mode:\footnote{A similar formula arises for detecting 1 photon in each of $k$ distinct modes and 0 photons in the remaining $m-k$ modes.}
$$\bra{1}^{\otimes m} \( \varphi(U) \(\bigotimes_{i = 1}^m \rho_i\) \varphi(U)^\dag \)\ket{1}^{\otimes m} = \frac{\Per(A)}{\prod_{i=1}^m(1 +\left<n_i\right>)}.$$

One might then reasonably ask, ``how hard is it to compute such probabilities?''  The following theorem answers that question in the exact case.

\begin{theorem}
\label{thm:psd_hard}
The permanent of a positive-definite matrix in $\mathbb{Z}^{n \times n}$ is $\sharpP$-hard.  This implies $\sharpP$-hardness for the larger class of positive \emph{semidefinite} matrices.
\end{theorem}
\begin{proof}
It is well-known that the permanent of a 0-1 matrix is $\sharpP$-hard \cite{valiant}.  Therefore, let $B \in \{0,1\}^{n \times n}$ and consider the matrix
$$\Lambda_B = \begin{pmatrix}
0 & B \\
B^T & 0
\end{pmatrix}$$
Since $\Per(B) \ge 0$, we have $\Per(B) = \sqrt{\Per(\Lambda_B)}$.  Also observe that $\Lambda_B^T = \Lambda_B$, so $\Lambda_B$ is Hermitian, that is, diagonalizable with real eigenvalues.  Furthermore, since $B$ is a 0-1 matrix, its spectral radius is at most $2n$.  Defining $\Lambda_B(x) := \Lambda_B + x I$, we see that $\Lambda_B(x)$ is positive-definite for all $x > 2n$.  

Notice now that $\Per(\Lambda_B(x))$ is a degree-$2n$ polynomial in $x$.  Therefore, given an oracle that calculates the permanent of a positive-definite matrix, we can interpolate a monic polynomial through the points $x = 2n+1, 2n+2, \ldots, 4n$ to recover the polynomial $\Per(\Lambda_B(x))$.  Since $\Per(\Lambda_B(0)) = \Per(\Lambda_B)$, the permanent of a positive-definite matrix under \emph{Turing reductions} is $\sharpP$-hard.

We now only have left to prove that the above reduction can be condensed into a single call to the positive-definite matrix permanent oracle.  Since the matrix $B$ is a 0-1 matrix, the polynomial $\Per(\Lambda_B(x))$ has positive integer coefficients, the largest of which is at most $(2n)!$.  Therefore, if $x > (2n)!$, then we can simply read out each of the coefficients of  $\Per(\Lambda_B(x))$ with a single oracle call.  Clearly, this requires at most a polynomial increase in the bit length of the integers used in the reduction.
\end{proof}

Theorem~\ref{thm:psd_hard} implies that there is some linear optical experiment one can perform with thermal input states for which calculating the exact success probability is computationally difficult.  We would like to say that this also precludes an efficient classical sampling algorithm (unless $\PH$ collapses), as is done in work by Aaronson and Arkhipov \cite{aark} and Bremner, Jozsa, Shepherd \cite{bjs}.  Unfortunately, those arguments rely on the fact that even finding an \emph{approximation} to their output probabilities is difficult, but the following theorem heavily suggests that such a result cannot exist.

\begin{theorem}[Rahimi-Keshari, Lund, Ralph \cite{rahimi:2015}]
\label{thm:psd_approx_easy}
There exists an efficient classical sampling algorithm for Boson Sampling with thermal input states.  Furthermore, multiplicatively approximating the permanent of a PSD matrix is in the class $\mathsf{FBPP}^\NP$ .
\end{theorem}

Intuitively, such an algorithm exists because it is possible to write the permanent of a PSD matrix as an integral\footnote{Suppose we have PSD matrix $A = CC^\dag$ where $C = \{c_{i,j}\}$.  Then the permanent of $A$ can be expressed as the following expected value over complex Gaussians:
$$\Per(A) = \mathlarger{\mathlarger{\mathlarger{\operatorname*{ \mathbb E}}}}_{x \in \mathcal G_{\mathbb C}(0,1)^n}\left[ \prod_{i=1}^n \left| \sum_{j=1}^n c_{i,j} x_j \right|^2\right].$$} of a nonnegative function, on which we can use Stockmeyer's approximate counting algorithm \cite{stockmeyer}.  Such a representation as a sum of positive terms also implies that the permanent of a PSD matrix is nonnegative. 

Notice that this also justifies our use of techniques distinct from the linear optical approach.  Suppose we can encode the answer to a $\GapP$-hard problem into the permanent of a PSD matrix as we do with real orthogonal matrices, then multiplicatively approximating the permanent of a PSD matrix would also be $\GapP$-hard under Turing reductions (see Theorem~\ref{thm:approx1} in Appendix~\ref{sec:approx}).  On the other hand, Theorem~\ref{thm:psd_approx_easy} says that such a multiplicative approximation \emph{does} exist, so
$$\PH \subseteq \P^\GapP \subseteq\mathsf{BPP}^\NP \subseteq \mathsf{\Sigma}_3^\P.$$
Therefore, either such a reduction does not exist or the polynomial hierarchy collapses to the third level.


\subsection{More Permanent Consequences of the Main Result}

In this section, we try to give a sense in which our proof for the hardness of the permanent for real orthogonal matrices leads to new hardness results for many classes of matrices.  The structure of this section is as follows:  we will first \emph{restrict} as much as possible the class of matrices for which the permanent is $\sharpP$-hard; we will then observe that the permanent for any larger class of matrices must also be hard, which will show hardness for many natural classes of matrices.

We call matrix $A$ an \emph{involution} if $A = A^{-1}$.

\begin{theorem}
\label{thm:orthogonal_involution_hard}
Let $A$ be a real orthogonal involution with $\Per(A) \ge 0$.  The permanent of $A$ is $\sharpP$-hard.
\end{theorem}
\begin{proof}
Let $C : \{0,1\}^n \rightarrow \{0,1\}$ be a Boolean function for which we want to calculate $\Delta_C$.  We will construct a new circuit $C' : \{0,1\}^{n+1} \rightarrow \{0,1\}$ such that for $x \in \{0,1\}^n$ and $b \in \{0,1\}$ we have $C'(x, b) = C(x) \vee b.$  It is not hard to see then that $\Delta_{C'} = \Delta_C + 2^n$.  Importantly, this implies that $\Delta_{C'} \ge 0$.  

Now let us leverage the reduction in Theorem~\ref{thm:main_real_orthogonal} to build a real orthogonal matrix $B$ such that $\Per(B) \propto \Delta_{C'}$.  As in the proof of Theorem~\ref{thm:psd_hard}, let
$$\Lambda_B = \begin{pmatrix}
0 & B \\
B^T & 0
\end{pmatrix}.$$
Since $\Delta_{C'} \ge 0$, we have $\Per(B) \ge 0$, which implies that $\Per(B) = \sqrt{\Per(\Lambda_B)}$.  However, since $B$ is orthogonal, we have that $\Lambda_B^2 = I$, so $\Lambda_B$ is an involution.   Furthermore, $\Lambda_B = \Lambda_B^T$, so $\Lambda_B$ is a real orthogonal matrix.  Therefore, the permanent of real orthogonal involutions is $\sharpP$-hard.
\end{proof}
	
We call a matrix $A$ \emph{special} if $\det(A) = 1$.  Furthermore, a matrix $A$ is \emph{symplectic} if $A^T\Omega A = \Omega \text{ where } \Omega = \(\begin{smallmatrix}0 & I_n \\ -I_n & 0\end{smallmatrix}\)$.  We strengthen Theorem~\ref{thm:orthogonal_involution_hard} to provide the smallest class of matrices for which we know the permanent is $\sharpP$-hard.

\begin{theorem}
\label{thm:special_orthogonal_involution_hard}
Let $A$ be a real special orthogonal symplectic involution with $\Per(A) \ge 0$.  The permanent of $A$ is $\sharpP$-hard.\end{theorem}
\begin{proof}
Let $B$ be a real orthogonal involution, and let $I_n$ be the $n \times n$ identity matrix. Consider the matrix 
$$I_2 \otimes B = \begin{pmatrix}
B & 0 \\
0 & B
\end{pmatrix}.$$
Notice that 
$$\det(I_2 \otimes B) = \det(B)^2 = \det(B^2) = \det(I_n) = 1,$$
where we use that $B^2 = I_n$ is an involution for the third equality. Therefore, $I \otimes B$ is special.  It is also easy to verify that $I_2 \otimes B$ is real orthogonal symplectic involution.  Assuming $\Per(B) \ge 0$, we have $\Per(B) = \sqrt{\Per(I_2 \otimes B)}$.  Combining the above with Theorem~\ref{thm:orthogonal_involution_hard}, we get that the permanent of real special orthogonal involutions is $\sharpP$-hard.
\end{proof}

Since the set of $n \times n$ real special orthogonal matrices form a group $\mathrm{SO}(n, \mathbb{R})$, we immediately get $\sharpP$-hardness for all the matrix groups containing it.  
\begin{corollary}
\label{cor:classical_groups}
The permanent of an $n \times n$ matrix $A$ in any of the classical Lie groups over the complex numbers is $\sharpP$-hard.  That is, it is hard for the following matrix groups:
\begin{align*}
 \textbf{General linear:  }& A \in \mathrm{GL}(n) \text{ iff } \det(A) \neq 0 \\
 \textbf{Special linear:  }& A \in \mathrm{SL}(n) \text{ iff } \det(A) = 1 \\
 \textbf{Orthogonal:  }& A \in \mathrm{O}(n) \text{ iff } AA^T = I_n \\
 \textbf{Special orthogonal:  }& A \in \mathrm{SO}(n) \text{ iff } AA^T = I_n \text{ and } \det(A) = 1 \\
 \textbf{Unitary:  }& A \in \mathrm{U}(n) \text{ iff } AA^\dag = I_n \\ 
 \textbf{Special unitary:  }& A \in \mathrm{SU}(n) \text{ iff } AA^\dag = I_n \text{ and } \det(A) = 1 \\
 \textbf{Symplectic:  }& A \in \mathrm{Sp}(2n) \text{ iff } A^T\Omega A = \Omega \text{ where } \Omega = \(\begin{smallmatrix}0 & I_n \\ -I_n & 0\end{smallmatrix}\)
\end{align*}
\end{corollary}
\begin{proof}
Since $\mathrm{SO}(n, \mathbb{R})$ is a subgroup of all the stated Lie groups besides the symplectic group $\mathrm{Sp}(2n)$, their permanents are $\sharpP$-hard by Theorem~\ref{thm:special_orthogonal_involution_hard}.  Theorem~\ref{thm:special_orthogonal_involution_hard} handles the symplectic case separately.
\end{proof}

%
%
%


\section{Open Problems}
\label{sec:open_problems}

This paper gives many new classes of matrices for which the permanent is hard.  Nevertheless, there exist classes of matrices which have unknown permanent complexity, and proving $\sharpP$-hardness or otherwise remains a central open problem.  For instance, is computing the permanent of an orthogonal matrix modulo a prime $p$ hard for all $p \neq 2,3$?  Notice that our result only gives $\ModpP$-hardness for 1/16th of all primes.


Another interesting open question about permanents concerns the complexity of multiplicatively approximating permanents of PSD matrices.  Although we show the exact version of this problem to be $\sharpP$-hard in this paper, we know that an $\mathsf{FBPP}^\NP$ algorithm exists \cite{rahimi:2015}.  Could this problem actually just be in $\P$?  Is there any more insight to be gained by viewing PSD permanents as probabilities of certain boson sampling experiments?  For instance, Chakhmakhchyan, Cerf, and Garcia-Patron \cite{chakhmakhchyan:2016} have recently detailed conditions on the eigenvalues of a PSD matrix for which a linear optical sampling algorithm gives a better \emph{additive} approximation to the permanent than the classic approximation algorithm of Gurvits \cite{gurvits:2005}.

\section{Acknowledgments}
We would like to thank Scott Aaronson for posing the question which led to this paper and for his comments on this paper.  We would also like to thank Rio LaVigne and Michael Cohen for some key mathematical insights.

\addcontentsline{toc}{section}{References}
\bibliographystyle{plain}
\bibliography{bibliography}

\appendix
\section{Counting Classes}
\label{sec:complexity}

Let us introduce the complexity classes we use in this paper. Note that the permanent is a function, so computing it is a function problem. Hence, we will sometimes need the class $\FP$ to stand in for $\P$ when we are talking about function problems.  
\begin{definition}
$\FP$ is the class of functions computable by deterministic Turing machines in polynomial time. 
\end{definition}
Of course, computing the permanent is, in general, thought to be intractable (i.e., not in $\FP$). We use a variety of different classes to capture the difficulty of computing the permanent (depending on the kind of matrix, underlying field, etc.), but the most important class is $\sharpP$: 
\begin{definition}
$\sharpP$ is the class of function problems of the form ``compute the number of accepting paths of a polynomial-time non-deterministic Turing machine." For example, given a classical circuit of NAND gates as input, the problem of computing the number of satisfying assignments is in $\sharpP$ (and indeed, is $\sharpP$-complete). 
\end{definition}
Since $\sharpP$ is a class of function problems (more specifically, counting problems), we often consider $\P^{\sharpP}$ to compare $\sharpP$ to decision classes. Observe that $\P^{\sharpP} = \P^{\PP}$ since, on the one hand, the $\sharpP$ oracle can count paths to simulate $\PP$, and on the other hand, we can use the $\PP$ oracle to binary search (on the number of accepting paths) to count exactly. We add that $\P^{\sharpP} \subseteq \PSPACE$ is a upper bound for $\sharpP$, and Toda's theorem \cite{toda} gives $\PH \subseteq \P^{\sharpP}$.

Fenner, Fortnow, and Kurtz \cite{fennergap:1991} define a very closely related class, $\GapP$, which is also relevant to us. 
\begin{definition}
$\GapP$ is the class of function problems of the form ``compute the number of accepting paths \emph{minus} the number of rejecting paths of a polynomial-time non-deterministic Turing machine."
\end{definition}
We have $\GapP \supseteq \sharpP$ since we can take a $\sharpP$ problem (manifest as a non-deterministic Turing machine) and at the end of each rejecting path, add a non-deterministic branch which accepts in one half and rejects in the other. In the other direction, any $\GapP$ problem can be solved with at most two calls to a $\sharpP$ oracle (one for accepting paths, one for rejecting), and a subtraction. Hence, for most of our results we neglect the difference and stated $\sharpP$-hardness. 

Nonetheless, $\GapP$ and $\sharpP$ \emph{are} different. For one, functions in $\sharpP$ are non-negative (and integral) by definition, whereas functions in $\GapP$ can take negative values. The distinction is also important in the context of approximation; Stockmeyer's approximate counting gives a multiplicative approximation to any $\sharpP$ problem in $\BPP^{\NP}$, whereas it is known that multiplicative approximation to a $\GapP$-hard problem remains $\GapP$-hard under Turing reductions (see Theorem~\ref{thm:approx1}). 

One cannot even get very bad multiplicative approximations to $\GapP$-hard problems. Even the worst multiplicative approximation will distinguish zero from non-zero outputs, and this problem is captured by the class $\CeP$, defined below. 
\begin{definition}
$\CeP$ is the class of decision problems of solvable by a non-deterministic polynomial-time machine which accepts if it has the same number of accepting paths as rejecting paths. 
\end{definition}
A good upper bound for $\CeP$ is simply $\PP$. This is easily seen once we have the following theorem.
\begin{theorem}
\label{thm:gap_arithmetic}
Suppose $f_1, f_2 \in \Sigma^{*} \rightarrow \mathbb Z$ are functions computable in $\GapP$. Then $f_1 + f_2$, $-f_1$, and $f_1 f_2$ are computable in $\GapP$. 
\end{theorem}
\begin{proof}
Let $M_1$ and $M_2$ be non-deterministic machines witnessing $f_1 \in \GapP$ and $f_2 \in \GapP$ respectively. Then the machines for $f_1 + f_2$, $-f_1$, and $f_1 f_2$ are defined as follows. 
\begin{enumerate}
\item For $f_1 + f_2$, non-deterministically branch at the start, then run $M_1$ in one branch and $M_2$ in the other. 
\item For $-f_1$, take the complement of $M_1$. That is, make every accepting path reject, and make every rejecting path accept. 
\item For $f_1 f_2$, run $M_1$ to completion, then run $M_2$ to completion (in every branch of $M_1$). Accept if the two machines produce the same outcome, otherwise reject. 
\end{enumerate} 
The last construction may require some explanation. Let $a_1, a_2$ be the number of accepting paths of $M_1$ and $M_2$ respectively, and similarly let $b_1, b_2$ be the numbers of rejecting paths. Then there are $a_1 a_2 + b_1 b_2$ accepting paths for the new machine and $a_1 b_2 + a_2 b_1$ rejecting paths, so as a $\GapP$ machine it computes 
$$
a_1 a_2 - a_1 b_2 - a_2 b_1 + b_1 b_2 = (a_1 - b_1) (a_2 - b_2) = f_1(x) f_2(x).
$$
\end{proof}
Theorem~\ref{thm:gap_arithmetic} implies that $\CeP \subseteq \PP$ because we can square and negate the gap. In other words, we can find a machine such that the gap is always negative (i.e., strictly less than half of all paths accept) unless the original machine had gap zero, in which case the gap is still zero (or, WLOG, very slightly positive). It is also worth noting that $\coCeP$ is known to equal $\NQP$, by a result of Fenner et al.\ \cite{NQP:1998}. 


\begin{definition}
The class $\NQP$ contains decision problems solvable by a polynomial-time quantum Turing machine (or, equivalently, a uniform, polynomial-size family of quantum circuits) where we accept if there is any nonzero amplitude on the accept state at the end of the computation. 
\end{definition}
Quantum classes with exact conditions on the amplitudes (e.g., $\NQP$ or $\EQP$) tend to be very sensitive to the gate set, or QTM transition amplitudes allowed. Adleman, Demarrais, and Huang \cite{quantumcomputability:1997} are careful to define $\NQP$ for the case where the transition amplitudes are algebraic and real. 


Finally, we specify computational hardness for our finite field problems using a mod $k$ decision version of $\sharpP$. 
\begin{definition}
For any integer $k \geq 2$, let $\ModkP$ be the class of decision problems solvable by a polynomial time non-deterministic machine which rejects if the number of accepting paths is divisible by $k$, and accepts otherwise. In the special case $k = 2$, $\ModkP$ is also known as ``parity $\P$", and denoted $\parityP$. 
\end{definition}
Clearly $\P^{\sharpP}$ is an upper bound for $\ModkP$.  We are finally ready to state the main hardness result for these counting classes, namely, the celebrated theorem of Toda \cite{toda} and a subsequent generalization by Toda and Ogiwara \cite{todaogiwara:1992}.  There are many important consequences of Toda's work, but we only require the following formulation.

\begin{theorem}[Toda's Theorem \cite{toda, todaogiwara:1992}]
Let $A$ be one of the counting classes $\ModkP$, $\CeP$, $\#P$, $\PP$, or $\GapP$. Then $\PH \subseteq \BPP^A$.
\end{theorem}

This means in particular that, if a problem is hard for any of these classes, then there is no efficient algorithm for the problem unless $\PH$ collapses.


\section{Real Construction of Toffoli (Proof of Lemma~\ref{lem:toff_proof})}
\label{app:toff_proof}

In this appendix we prove Lemma~\ref{lem:toff_proof} from Section~\ref{sec:real_orthogonal}.  Let us first define $R_\theta$ as the rotation by $\theta$ about the $Y$-axis.  That is, $R_\theta = \cos(\theta /2) I - i \sin(\theta /2)Y$ where $Y$ is the Pauli $\sigma_Y$ matrix.  For our purposes, we only require the following two matrices:
\begin{center}
\begin{minipage}{.5\linewidth}
\centering
$\B = \frac{1}{2}
\begin{pmatrix}
 \sqrt{2+\sqrt{2}} & -\sqrt{2-\sqrt{2}} \\
 \sqrt{2-\sqrt{2}} & \sqrt{2+\sqrt{2}} 
\end{pmatrix}$
\end{minipage}
\begin{minipage}{.3\linewidth}
\centering
$\R=
\begin{pmatrix}
0 & -1 \\
1 & 0  
\end{pmatrix}$
\end{minipage}
\end{center}

Let us now recall the statement of the lemma:
\begin{repeated}
There exists a circuit of CSIGN, Hadamard, and $\B$ gates which implements a Toffoli gate \emph{exactly}.
\end{repeated}
\begin{proof}

We construct the Toffoli gate from the CSIGN, Hadamard, and $\B$ gates in three steps:
\begin{enumerate}
\item \textbf{Construct a controlled-controlled-$\R$ gate ($\CCR$) from $\CSIGN$ and $\B$ gates.}  $\CCR$ is a three-qubit gate that applies $\R$ to the third qubit if the first two qubits are in the state $\ket{11}$.  Notice that $\CCR$ is already a kind of ``poor man's''  Toffoli gate.  If it were not for the minus sign in the $\R$ gate, we would be done.  The construction is given in Figure~\ref{fig:B_to_CCR}.  Observe that if either of the two control qubits is zero, then any CNOT gate controlled by that qubit can be ignored.  The remaining gates will clearly cancel to the identity.  Furthermore, if the two control qubits are in the state $\ket{11}$, then on the last qubit, we apply the operation $X \B^{-1} X \B X \B^{-1} X \B$.  Since $X \B^{-1} X = \B$,  
$$X \B^{-1} X \B X \B^{-1} X \B =  \B^4 = \R.$$

Notice that this construction uses $\CNOT$ gates, but observe that a $\CNOT$ is a CSIGN gate conjugated by the Hadamard gate: $$(I \otimes H) \CSIGN (I \otimes H) = \CNOT.$$

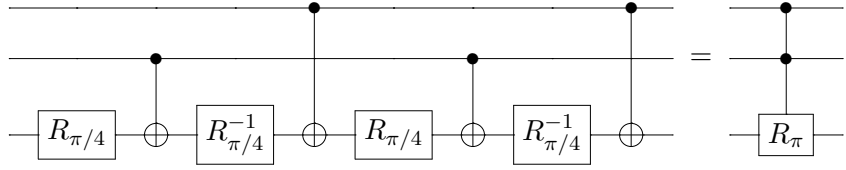
\begin{figure}[ht]
\centering
\mbox{
\Qcircuit @C=1em @R=1.5em {
& \qw & \qw & \qw & \ctrl{2} & \qw & \qw & \qw & \ctrl{2} & \qw & & & \ctrl{1} & \qw \\ 
& \qw & \ctrl{1} & \qw & \qw & \qw & \ctrl{1} & \qw & \qw & \qw & = &  & \ctrl{1} & \qw \\
& \gate{\B} & \targ & \gate{\B^{-1}} & \targ & \gate{\B} & \targ & \gate{\B^{-1}} & \targ & \qw & & & \gate{\R} & \qw
}}
\caption{Generating $\CCR$ from the CNOT and $\B$ gates.}
\label{fig:B_to_CCR}
\end{figure}

\item \textbf{Construct a non-affine classical reversible gate from $\CSIGN$ and $\CCR$ gates.} By classical, we simply mean that the gate maps each computational basis state to another computational basis state (i.e., states of the form $\ket{x}$ for $x \in \{0,1\}^n$).  If this transformation is non-affine, then it suffices to generate Toffoli (perhaps with some additional ancilla qubits) by Aaronson et al. \cite{ags:2015}. The construction is shown in Figure~\ref{fig:CCR_to_classical_nonaffine}.

\begin{figure}[ht]
\centering
\mbox{
\Qcircuit @C=1.5em @R=1.5em {
& \gate{\R} & \ctrl{1} & \qw & \ctrl{1} &\qw & & & \targ & \ctrl{1}  & \qw \\ 
& \ctrl{-1} & \ctrl{1} & \ctrl{1} & \ctrl{0} & \qw & = & & \ctrl{-1} & \ctrl{1}  & \qw \\
& \ctrl{-1} & \gate{\R} & \ctrl{0} & \qw & \qw & & & \ctrl{-1} & \targ  & \qw 
}}
\caption{Generating non-affine classical gate from $\CCR$ and $\CSIGN$.}
\label{fig:CCR_to_classical_nonaffine}
\end{figure}
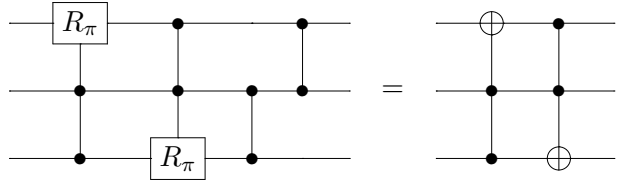

\item \textbf{Use the non-affine gate to generate Toffoli.}  We give an explicit construction in Figure~\ref{fig:nonaffine_to_toffoli}.  Notice that the fourth qubit is an ancillary qubit starting in the $\ket{0}$ state.\footnote{Indeed, this ancillary qubit is necessary because the non-affine gate in Figure~\ref{fig:CCR_to_classical_nonaffine} is an even permutation and the Toffoli gate is an odd permutation on three bits.}
\end{enumerate}
 
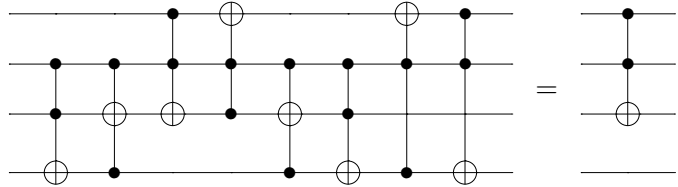
\begin{figure}[ht]
\centering
\mbox{
\Qcircuit @C=1.2em @R=1.2em {
& \qw & \qw & \ctrl{1} & \targ & \qw & \qw & \targ & \ctrl{1} & \qw & & & \ctrl{1} & \qw \\
& \ctrl{1} & \ctrl{1} & \ctrl{1} & \ctrl{-1} & \ctrl{1} & \ctrl{1} & \ctrl{-1} & \ctrl{2} & \qw & \raisebox{-2em}{=} & & \ctrl{1} & \qw \\
& \ctrl{1} & \targ & \targ & \ctrl{-1} & \targ & \ctrl{1} & \qw & \qw & \qw & & & \targ & \qw \\
 & \targ & \ctrl{-1} & \qw & \qw & \ctrl{-1} & \targ & \ctrl{-2} & \targ & \qw & & & \qw & \qw 
}}
\caption{Generating Toffoli gate from non-affine gate in Figure~\ref{fig:CCR_to_classical_nonaffine}.}
\label{fig:nonaffine_to_toffoli}
\end{figure}
\end{proof}


\section{Gadget Details}
\label{app:fftech}

As discussed above in Section~\ref{sec:real_orthogonal} and Section~\ref{sec:finite_field}, our results on orthogonal matrices depend on a collection of gadgets. In the real orthogonal setting (Section~\ref{sec:real_orthogonal}), each gadget is a real orthogonal matrix with algebraic entries, and all entries have clear, compact expressions in terms of radicals. However, in Section~\ref{sec:finite_field}, we wish to reuse the same gadgets over finite fields, and radicals are no longer the best representation. 

Instead, we will show that our (real) gadget matrices have entries in $\mathbb Q(\alpha)$, the algebraic field extension of the rational numbers by $\alpha$, where $\alpha = \sqrt{2 + \sqrt{2}} + \sqrt{3 + \sqrt{6}} \approx 4.182173283$ is the largest real root of irreducible polynomial
$$
f(x) = x^{16} - 40x^{14} + 572x^{12} - 3736x^{10} + 11782 x^8 - 17816 x^6 + 11324 x^4 - 1832 x^2 + 1.
$$
More specifically, we will write every entry as a polynomial in $\alpha$, with rational coefficients and degree less than $16$. 

This is a cumbersome representation for hand calculation, but there are some advantages. First, it eliminates any ambiguity about, for instance, which square root of $2$ to use in a finite field. Second, we can check the various conditions our gadgets need to satisfy in the field $\mathbb Q(\alpha)$, and then argue that the verification generalizes to $\mathbb F_p(\alpha)$, with a few caveats. So, without further ado, we present polynomials for a set of reals which generate all the entries of our gadgets. 

\subsection{Gadget entries}
\label{app:fftech_entries}

Since $\pm \frac{1}{\sqrt{2}}$ are the only entries in $\decoder$, our decoder gadget, we show how to express those entries as polynomials in $\alpha$.
\small 
\begin{align*}
\frac{1}{\sqrt{2}} &= \frac{1}{11776} \left( \alpha^{14}-53 \alpha^{12}+1077 \alpha^{10}-10561 \alpha^8+51555 \alpha^6-115791 \alpha^4+95207 \alpha^2-8379 \right),
\end{align*}
\normalsize
For our encoder gadget $\encoder$, we also must also express $\frac{1}{\sqrt{3}}$ as an element in $\mathbb{Q}(\alpha)$. Note that $\frac{1}{\sqrt{6}}$ can be obtained as $\frac{1}{\sqrt{2}} \cdot \frac{1}{\sqrt{3}}$. 

\small
\begin{align*}
\frac{1}{\sqrt{3}} &= \frac{1}{11776} \left( \alpha^{14}-53 \alpha^{12}+1077 \alpha^{10}-10561 \alpha^8+51555 \alpha^6-115791 \alpha^4+95207 \alpha^2-8379 \right).
\end{align*}
\normalsize

Showing that the entries of the $\B$ gate are in $\mathbb{Q}(\alpha)$ requires the following: 
\begin{align*}
\sqrt{2 + \sqrt{2}} &= \frac{1}{5888} \big( -123 \alpha^{15}+4932 \alpha^{13}-70785 \alpha^{11}+464494 \alpha^9 \\
& \qquad \qquad -1470141 \alpha^7+2209176 \alpha^5-1357287 \alpha^3+193302 \alpha \big) \\
\sqrt{2 - \sqrt{2}} &= \frac{1}{5888} \big( 216 \alpha^{15}-8711 \alpha^{13}+126234 \alpha^{11}-841629 \alpha^9 \\ 
& \qquad \qquad +2733428 \alpha^7-4270353 \alpha^5+2799098 \alpha^3-466411 \alpha \big)
\end{align*}
\normalsize

Finally, we have the $V$ gate. We already have the $\frac{1}{3\sqrt{2}}$ in front, and the various multiples of $\sqrt{2}$ inside, so we just need $\sqrt{3 \pm \sqrt{6}}$ and $\sqrt{6 \pm 2 \sqrt{6}}$. These are related by a factor of $\sqrt{2}$, so it suffices to give $\sqrt{3 \pm \sqrt{6}}$. 
\begin{align*}
\sqrt{3 + \sqrt{6}} &= \frac{1}{5888} \big( 123 \alpha^{15}-4932 \alpha^{13}+70785 \alpha^{11}-464494 \alpha^9 \\
& \qquad \qquad +1470141 \alpha^7-2209176 \alpha^5+1357287 \alpha^3-187414 \alpha \big) \\
\sqrt{3 - \sqrt{6}} &= \frac{1}{256} \big( 15 \alpha^{15}-598 \alpha^{13}+8505 \alpha^{11}-55084 \alpha^9 \\
& \qquad \qquad +171665 \alpha^7-256518 \alpha^5+161671 \alpha^3-25624 \alpha \big)
\end{align*}
\normalsize

The numbers above, combined with $\frac{1}{2}$ and $\frac{1}{3}$, generate all the entries of our real orthogonal gadgets. Note that the denominators in front of the polynomials above (e.g., $11776$, $5888$, $256$, $3$, etc.) all divide $35328 = 2^9 \cdot 3 \cdot 23$. In other words, this representation is a bad choice for fields of characteristic $2$, $3$, or $23$ because, in those cases, division by $35328$ is division by $0$. Aside from this restriction, the representation is well-defined for any field containing some root $\alpha$ of the polynomial $p$. 


We should not be surprised that the representation fails for fields of characteristic $2$ or $3$ because our matrices contain, for instance, the entries $\frac{1}{3}$ and $\frac{1}{\sqrt{2}}$. We also know the permanent of an orthogonal matrix is easy to compute in fields of characteristic $2$ or $3$, so it is actually no surprise to find this obstacle to our hardness proof. 

On the other hand, we can find no explanation for the requirement $p \neq 23$; it appears to be a quirk of the algebraic number $\alpha$. In fact, a different choice fails for different primes. Consider $\beta \approx 5.596386846$, the largest real root of 
\begin{align*}
x^{16}-56 x^{14}-32 x^{13}+1084 x^{12}+960 x^{11}-9224 x^{10}-8928 x^9+37702 x^8+ & \\
33920 x^7-73736 x^6-53216 x^5+63932 x^4+23488 x^3-21560 x^2+3808 x-191 &.
\end{align*}
This appendix is long enough without doing all the same steps for $\beta$, so let us claim without proof that $\mathbb Q(\beta) = \mathbb Q(\alpha)$. Furthermore, when we represent the matrix entries as polynomials in $\beta$ (we omit the details), the denominators prohibit the use of this representation for fields of characteristic $2$, $3$, $191$, and $3313$, but \emph{not} $23$. Hence, for all primes $p$ other than $2$ or $3$, there is some representation that works for that prime.

%
%
%

\subsection{Galois Extension}
\label{app:fftech_galois}

We need $\mathbb Q(\alpha)$ to be a Galois extension to apply Chebotarev's theorem, which we use to prove Theorem~\ref{thm:ffmain2}. Another helpful consequence is that if $\alpha$ is in some field, then all the roots of $f$ are also in the field since they can be expressed as polynomials in $\alpha$.  

The most direct way to prove $\mathbb Q(\alpha)$ is a Galois extension is to write all 16 roots of $f$ in terms of $\alpha$. Since $f$ is an even polynomial, half of the roots are just the negatives of the other half, so we restrict our attention to the $8$ positive roots. 

\begin{center}
\begin{tabular}{|r|c|}
\hline
Root & Polynomial \\
\hline
$0.0234$ & \scriptsize $\frac{1}{5888}\(-129 \alpha^{15}+5043 \alpha^{13}-69381 \alpha^{11}+425303 \alpha^9-1214867 \alpha^7+1629561 \alpha^5-919335 \alpha^3+122941 \alpha\)$ \\
$0.4866$ & \scriptsize $\frac{1}{2944}\(123 \alpha^{15}-4932 \alpha^{13}+70785 \alpha^{11}-464494 \alpha^9+1470141 \alpha^7-2209176 \alpha^5+1357287 \alpha^3-190358 \alpha\)$ \\
$1.1057$ & \scriptsize $\frac{1}{2944}\(-234 \alpha^{15}+9343 \alpha^{13}-133200 \alpha^{11}+865713 \alpha^9-2709218 \alpha^7+4054545 \alpha^5-2537860 \alpha^3+391327 \alpha\)$ \\
$1.5073$ & \scriptsize $\frac{1}{5888}\(561 \alpha^{15}-22465 \alpha^{13}+321849 \alpha^{11}-2108561 \alpha^9+6681723 \alpha^7-10170267 \alpha^5+6517531 \alpha^3-1055763 \alpha\)$ \\
$1.5690$ & \scriptsize $\frac{1}{5888}\(-93 \alpha^{15}+3779 \alpha^{13}-55449 \alpha^{11}+377135 \alpha^9-1263287 \alpha^7+2061177 \alpha^5-1441811 \alpha^3+278997 \alpha\)$ \\
$2.5897$ & \scriptsize $\frac{1}{2944}\(111 \alpha^{15}-4411 \alpha^{13}+62415 \alpha^{11}-401219 \alpha^9+1239077 \alpha^7-1845369 \alpha^5+1180573 \alpha^3-198025 \alpha\)$ \\
$3.0997$ & \scriptsize $\frac{1}{5888}\(339 \alpha^{15}-13643 \alpha^{13}+197019 \alpha^{11}-1306123 \alpha^9+4203569 \alpha^7-6479529 \alpha^5+4156385 \alpha^3-653825 \alpha\)$ \\
$4.1821$ & \scriptsize $\alpha$ \\
\hline
\end{tabular}
\end{center}


\section{Approximation}
\label{sec:approx}

Much like in Aaronson's paper \cite{aar:per}, our hardness reductions for \emph{exactly} computing the permanent lead naturally to hardness of approximation results as well.  Approximation results comes in two flavors: additive and multiplicative. For example, Gurvits' algorithm \cite{gurvits:2005} approximates the permanent of a matrix $A$ up to $\pm \varepsilon \| A \|^n$ additive error. We will focus strictly on multiplicative approximation. That is, the result of the approximation should be between $\frac{1}{k} \Per(A)$ and $k \Per(A)$ for some $k$. 

We give approximation results only for real orthogonal matrices since it is unclear how to even define multiplicative approximation in a finite field. All of our results follow from the fact that we actually prove $\GapP$-hardness (since we compute the gap, $\Delta_C$, rather than just the number of satisfying assignments). None of the results use anything specific to permanents; they are all $\GapP$ folklore, but we state them as permanent results for clarity. 

\begin{theorem}
\label{thm:approx1}
Suppose $A$ is an oracle that approximates the permanent of a real orthogonal matrix to \emph{any} multiplicative factor. In other words, $A$ is an oracle for the sign (zero, positive, or negative) of the permanent. Then $\GapP \subseteq \FP^{A}$.
\end{theorem}
\begin{proof}
We give an $\FP^{A}$ algorithm for computing $\Delta_C$ for a classical circuit $C$. Since this problem is $\GapP$-hard, we get $\GapP \subseteq \FP^{A}$.  

By earlier results, we can construct a real orthogonal matrix with permanent proportional to $\Delta_C$. Then we can apply the oracle to compute the sign of the permanent, and hence the sign of $\Delta_C$. This is helpful, but we can do better. 

Recall that we can add or subtract two $\GapP$ functions (see Appendix~\ref{sec:complexity}), so for any integer $k$, we can construct a circuit $C_k$ such that $\Delta_{C_k} = \Delta_C - k$. Then we can apply $A$ to give us the sign of $\Delta_{C_k}$, or equivalently, compare $\Delta_C$ to $k$. In other words, we can use $A$ to binary search for the value of $\Delta_C$, which we know to be an integer in the range $-2^n$ and $2^n$.
\end{proof}

Recall that $\CeP$ is the class of decision problems of solvable by a non-deterministic polynomial-time machine which accepts if it has the same number of accepting paths as rejecting paths. By Toda's theorem, $\PH \subseteq \BPP^{\CeP}$.

\begin{theorem}
\label{thm:approx2}
Suppose $A$ is an oracle that approximates the absolute value of permanent of a real orthogonal matrix to \emph{any} multiplicative factor. That is, $A$ tells us whether the permanent is zero. Then $\P^{\CeP} \subseteq \P^{A}$. 
\end{theorem}
\begin{proof}
The problem of computing whether $\Delta_C = 0$ for a classical circuit $C$ is $\CeP$-hard. But clearly we can construct a real, orthogonal matrix from the circuit with permanent proportional to $\Delta_C$, and then apply $A$ to determine if the permanent is zero, and hence whether $\Delta_C$ is zero. Therefore $\P^{\CeP} \subseteq \P^{A}$. 
\end{proof}

Finally, we show that even a very poor approximation to the absolute value of the permanent still allows us to calculate the exact value of the permanent via a boosting argument.

\begin{theorem}
\label{thm:approx3}
Suppose $A$ is an oracle that approximates the absolute value of the permanent of an $n \times n$ real orthogonal matrix to within a $2^{n^{1-\varepsilon}}$ factor for some $\varepsilon > 0$. Then $\GapP \subseteq \FP^{A}$.
\end{theorem}
\begin{proof}
We give an $\FP^{A}$ algorithm for computing $\Delta_C$ of a classical circuit. Since this problem is $\GapP$-hard, we get $\GapP \subseteq \FP^{A}$.  

As in Theorem~\ref{thm:approx1}, we can construct a circuit $C_k$ such that $\Delta_{C_k} = \Delta_C - k$ for any integer $k$. By applying oracle $A$ to the real orthogonal matrix corresponding to $C_k$, we can get a multiplicative estimate for $|\Delta_C - k|$. Let us assume for the moment that $A$ gives a multiplicative approximation to within a factor of $2$, and improve this to $2^{n^{1-\varepsilon}}$ later. 

Suppose we are given an interval $[a, b]$ guaranteed to contain $\Delta_C$. For instance, $\Delta_C$ is initially in $[-2^{n}, 2^{n}]$. Apply $A$ to find an estimate for $\Delta_{C_a} = \Delta_C - a$. Suppose the approximation we get is $x^{*}$. Then we have 
$$
a + \frac{1}{2} x^{*} \leq \Delta_C \leq a + 2 x^{*}.
$$
So $\Delta_C$ is in the interval $[a + \frac{1}{2} x^*, a + 2x^{*}] \cap [a,b]$. One can show that this interval is longest when $a + 2x^{*} = b$, where it has length $\frac{3}{4} (b-a)$. Since the interval length decreases by a constant factor each step, we only need $O(n)$ steps to shrink it from $[-2^n, 2^n]$ to length $<1$, and determine $\Delta_C$.

Finally, suppose we are given an oracle which gives an approximation to within a multiplicative factor $2^{n^{1-\varepsilon}}$. Theorem~\ref{thm:gap_arithmetic} in Appendix~\ref{sec:complexity} lets us construct a circuit $C^{m}$ (not to be confused with $C_k$) such that $\Delta_{C^m} = (\Delta_C)^m$. The circuit is essentially $m$ copies of $C$, so we can only afford to do this for $k$ polynomial in the size of $C$, otherwise our algorithm is too slow.

The point of $C^m$ is that a factor $\beta$ approximation to $\Delta_{C^m}$ gives a factor $\beta^{1/m}$ approximation of $\Delta_C$ by taking $m$th roots. This is excellent for reducing a constant approximation factor, but when $\beta$ grows with $n$, we must account for the fact that the size of $C^{m}$ grows with $n$ as well. In particular, the size of $C^{m}$ scales with $m$, and the dimension of the matrix in our construction scales linearly with $m$ as well. 

So, for our algorithm to succeed, we need $\beta(nm)^{1/m} \leq 2$ or 
$$\beta(nm) \leq 2^{m}$$
for $m$ a polynomial in $n$. Suppose we can afford $m = n^{c}$ copies of $C$. Then we succeed when $\beta(n^{1+c}) \leq 2^{n^c}$, or 
$$
\beta(n) \leq 2^{n^{1-\frac{1}{c+1}}}.
$$
Within the scope of polynomial time algorithms, we can make $\frac{1}{c+1}$ less than any $\varepsilon$, and thereby handle any $2^{n^{1-\varepsilon}}$ approximation factor. 
\end{proof}

The core ideas in both Theorem~\ref{thm:approx1} and \ref{thm:approx3} were already noticed by Aaronson \cite{aar:per}, but we give slightly better error bounds for the latter theorem.



\section{Orthogonal Matrices mod 97 are \texorpdfstring{$\sharpP$-hard}{\#P-hard} via \texorpdfstring{$\NS$}{NS1}-approach}
\label{app:ns_approach}

It is natural to ask whether Theorem~\ref{thm:ffmain2} can be extended to more primes, or all primes. In other words, is there some prime $p \neq 2, 3$ such that it is easy to compute the permanent modulo $p$, even though computing the permanent over $\mathbb F_{p^4}$ is hard? In this appendix, we present a different construction for CSIGN gates (in fact, the construction originally used by KLM) which works in $\mathbb F_{97}$, where the $V$ gate does not. We conclude that there is at least one more prime, namely $p = 97$, where the permanent is hard. 

The original KLM construction builds an CSIGN gate from what they call an $\NS$ gate, instead of directly using a $V$ gate. Logically, the $\NS$ gate acts on one mode and does nothing to $0$ or $1$ photon, but flips the sign for $2$ photons. The construction of CSIGN from $\NS$ is shown in Figure~\ref{fig:csign_from_ns}. If $\ket{1,1}$ is the input state, the Hadamard gate turns it into a linear combination of $\ket{2,0}$ and $\ket{0,2}$, which then change phase by the $\NS$ gate, and get recombined into $-\ket{1,1}$ by the Hadamard gate. Otherwise, there are not enough photons for the $\NS$ gates to do anything, and the Hadamard gates cancel, so the gate does nothing (as a CSIGN should). 

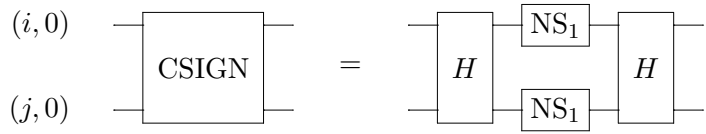
\begin{figure}[h]
\centering
\mbox{
\Qcircuit @C=1em @R=1.5em {
\lstick{(i, 0)} & & \multigate{1}{\CSIGN} &  \qw & &    & & &  \multigate{1}{H} & \gate{\NS} & \multigate{1}{H} & \qw \\ 
\lstick{(j, 0)} & & \ghost{\CSIGN} & \qw & & \raisebox{3em}{=} & & & \ghost{H} & \gate{\NS} &  \ghost{H} & \qw 
}}
\caption{Generating $\CSIGN$ from $H$ and $\NS$ \cite{klm}.}
\label{fig:csign_from_ns}
\end{figure}

It turns out it is impossible to construct an $\NS$ gate without at least two postselected modes, so the KLM $\NS$ is a three mode gate where the last two modes start and end (via postselection) in state $\ket{0,1}$. Unfortunately, the KLM $\NS$ gate postselects on a mode having zero photons, which is undesirable for our application. Therefore, we construct our own $\NS$ gate shown below. It postselects on the last two modes being $\ket{1,1}$ and has entirely real entries. 

The gate is 
$$\NS = \frac{1}{6}\begin{pmatrix}
 6-18 \gamma  & -\sqrt{6} \sqrt{9 \gamma -\sqrt{6-3 \gamma }-2} & -\sqrt{6} \sqrt{9 \gamma +\sqrt{6-3 \gamma }-2} \\
 -\sqrt{6} \sqrt{9 \gamma -\sqrt{6-3 \gamma }-2} & 9 \gamma +\sqrt{24-45 \gamma } & -3 \sqrt{2-4 \gamma } \\
 -\sqrt{6} \sqrt{9 \gamma +\sqrt{6-3 \gamma }-2} & -3 \sqrt{2-4 \gamma } & 9 \gamma -\sqrt{24-45 \gamma } \\
\end{pmatrix}
$$
where $\gamma \triangleq \frac{1}{18} \left(\sqrt{33}+3\right) \approx 0.4858090359$. 

One can verify the following identities hold. 
\begin{align*}
\bra{0,1,1} \phi(\NS) \ket{0,1,1} &= \gamma, \\
\bra{1,1,1} \phi(\NS) \ket{1,1,1} &= \gamma, \\
\bra{2,1,1} \phi(\NS) \ket{2,1,1} &= -\gamma.
\end{align*}
That is, with amplitude $\gamma$ the postselection succeeds, and the three mode gate behaves like an $\NS$ gate on the first mode. 

The field extension containing this gate is of higher degree than $\mathbb Q(\alpha)$, so we have not computed it explicitly. If we proved the equivalent of Theorem~\ref{thm:ffmain2} in that extension, we would expect the density to be worse. However, this construction of an CSIGN works for at least one prime where $V$ does not, namely $p = 97$.

\end{document}